 \documentclass[12pt, draftclsnofoot, onecolumn]{IEEEtran}
\usepackage{psfrag}

\ifCLASSINFOpdf
\else
\fi
\usepackage[cmex10]{amsmath}

\usepackage{mdwmath}
\usepackage{mdwtab}
\usepackage{color}
\usepackage{algorithm}
\usepackage{algorithmic}
\usepackage{graphicx,amsmath,amssymb,latexsym,subfig}
\usepackage{psfrag}
\usepackage{fancybox}
\usepackage{framed}
\usepackage{cite}
\usepackage{multirow}
\newtheorem{theorem}{Theorem}

\newtheorem{definition}{Definition}
\newtheorem{remark}{Remark}

\newtheorem{prop}{Proposition}

\begin{document}
%

\title{Finite Horizon Energy-Efficient Scheduling with Energy Harvesting Transmitters over Fading Channels}

\author{\IEEEauthorblockN{Baran Tan Bacinoglu\IEEEauthorrefmark{2}, Elif Uysal-Biyikoglu\IEEEauthorrefmark{1}, Can Emre Koksal\IEEEauthorrefmark{3}}
\\
\IEEEauthorblockA{\IEEEauthorrefmark{1}\IEEEauthorrefmark{2}METU, Ankara, Turkey,
\IEEEauthorrefmark{3} The Ohio State University\\
 E-mail:  barantan@metu.edu.tr, uelif@metu.edu.tr, koksal.2@osu.edu ,}

}


%


\maketitle

\begin{abstract}
In this paper, energy-efficient transmission schemes achieving maximal throughput over a finite time interval are studied in a problem setting including energy harvests, data arrivals and channel variation. The goal is to express the offline optimal policy in a way that facilitates a good online solution. We express any throughput maximizing energy efficient offline schedule (EE-TM-OFF) explicitly in terms of {\em water levels}. This allows per-slot real-time evaluation of transmit power and rate decisions, using estimates of the associated offline water levels. To compute the online power level, we construct a stochastic dynamic program that incorporates the offline optimal solution as a stochastic process. We introduce the ``Immediate Fill'' metric which provides a lower bound on the efficiency of any online policy with respect to the corresponding optimal offline solution. The online algorithms obtained this way exhibit performance close to the offline optimal, not only in the long run but also in short problem horizons, deeming them suitable for practical implementations.
\end{abstract}


%
\IEEEpeerreviewmaketitle

\section{Introduction}
\footnote{This paper is an extension of the study reported in ~\cite{6875018}.} 
\def\eg{{e.g.}}
\def\ie{{i.e.}}
Energy efficient packet scheduling with data arrival and deadline constraints has been the topic of numerous studies (e.g., \cite{UPE02,BeGa02,NuSr02,ZaMo09}). {\emph{Energy harvesting}} constraints have been incorporated in the recent years within these offline and online formulations (e.g., \cite{YaU2010,TuYe2010, MAAEUHE2010, Shroff2011, 5766183, Tassiulas2010, HoZang2010,5992841,6135979,6404770,6488477,6710085}.) A criticism that offline formulations often received is that the resulting offline policies did little to suggest good online policies. On the other hand, direct online formulations have been disconnected from offline formulations and the resulting policies (optimal policies or heuristics) have eluded explicit closed form expression as opposed to offline policies.

The problem of throughput maximization in energy harvesting communication systems and networks has also been widely studied and structural properties of throughput maximizing solutions have been investigated. For the throughput maximization problem in ~\cite{HoZang2010} and ~\cite{6202352}, it has been proved that the offline optimal solution can be expressed in terms of multiple distinct \emph{water levels} (to be made precise later in this paper) that are non-decreasing. In ~\cite{5992841}, this result is generalized for a continuous time system by introducing \emph{directional water-filling} interpretation of the offline solution. Similar results are also shown in  ~\cite{6135979}, ~\cite{6404770} and ~\cite{6488477} for the throughput maximization problem over fading channels with energy harvesting transmitters. 

Structural results on optimal adjustment of transmission rate/power according to energy harvesting processes naturally have duality relations with adapting to data arrival processes. Yet, few studies in the literature have addressed energy harvest and data arrival constraints simultaneously. To fill this gap, ~\cite{YaU2010} studied the offline solution that minimizes the transmission completion time where both packet arrivals and energy harvests occur during transmissions under static channel conditions. In~\cite{6354999}, the offline problem in~\cite{YaU2010} was extended to fading channels,  and in~\cite{Hakan} the broadcast channel with energy and data arrivals was considered, though the structure of the proposed solutions are not explicit and furthermore, do not provide much insight into in the derivation of online solutions.

{\emph{Asymptoticaly}} throughput optimal and delay optimal transmission policies were studied in~\cite{Sharma} under stochastic packet and energy arrivals. Online formulations of the energy harvesting scheduling problem based on dynamic programming were considered in~\cite{6202352}, and in~\cite{5992841}, which suggests online heuristics with reduced complexity. In~\cite{6176889}, the online solution maximizing overall throughput was formulated using a Markov Decision Process approach. The MDP approach was also used in~\cite{6334454} to obtain the performance limits of energy harvesting nodes with data and energy buffers. In~\cite{6485022}, a learning theoretic approach was employed to maximize long term (infinite horizon) throughput. The competitive ratio analysis was used in  ~\cite{6609112} for a throughput maximization problem on an energy harvisting channel with arbitrary channel variation and a simple online policy was shown to have a competitive ratio equal to the number of remaining time slots.  Recently, for an energy harvesting system with general  i.i.d. energy arrivals and finite size battery, an online power control policy ~\cite{7511344} was shown to maintain a constant-gap approximation to the optimal long term throughput average.

  The offline problem considering energy arrivals over fading channels has been studied in ~\cite{HoZang2010}, ~\cite{6202352}, ~\cite{5992841} and the offline problem considering energy and data arrivals over a static channel has been investigated in ~\cite{YaU2010}. To the best of our knowledge, the generic solution of the offline problem that considers energy and data arrivals together over time-varying channels has been covered exclusively in ~\cite{6875018} and this study is an extension of ~\cite{6875018}. Most significantly, this study presents an alternative approach to characterize the offline optimal solutions as opposed to an algorithm that iterates throughout the entire schedule rather than focusing on the optimal decision at a particular time slot. The offline solution introduced in this study, differs from existing solutions mainly in the construction of transmission schedules. We characterize the offline problem as that of finding optimal decisions successively in each time slot rather than searching for a complete transmission schedule. In our offline solution, we explicitly formulate the offline optimal decision at a given time slot and accordingly we can construct  offline schedules  slot by slot as if they are online schedules with known arrival patterns.  The offline optimal decisions, which are water levels individually set for each time slot, are expressed as explicit functions of  present energy and data buffer states, channel variations and future energy-data arrivals. In particular, the effect of channel variations are separated for each individual slot with the use of channel correction terms determining the optimal offine water level. This formulation of the offline solution allows us to characterize offline optimal decisions as random variables in a stochastic problem setting as in our online problem. 

The online problem we consider in this study  is formulated through stochastic dynamic programming which is also the typical approach taken by prior studies to express the online solution. On the other hand, different than existing dynamic programming formulations, our formulation of the online throughput maximization problem incorporates offline optimal schedules as stochastic processes that online optimal policies should follow closely and minimize expected regret due the variation of offline optimal decision in each successive online decision. Even tough the estimation of the offline optimal solution is not the primary goal of the online problem, good estimates of the offline optimal solution could capture the most of online optimal policies. In general, the dynamic programming solution suffers from the exponential complexity of the optimal solution as online decisions determine the future states of the system and highly depend on the time evolution of the system state in the optimal sense. In order to overcome this drawback, our online solution relies on offline optimal decisions which already consider future benefits in terms of energy-efficiency, however the cost-to-go function in our online formulation still carries an importance as the system might go to a state where offline optimal decisions can be better estimated.  In addition to this formulation of the online optimal solution, we introduce the immediate fill approach that lower bounds the ratio of expected performance of an online policy relative to the expected performance of offline optimal schedules and also suggests the maximization of the immediate fill metric in every slot to maximize this lower bound. Moreover, based on the offline optimal solution, we propose an online heuristic and through numerical analysis we show that this heuristic can  achieve average throughput rates close to the offline optimal performance even in the finite problem horizons for an arbirary numerical scenario.

\section{System Model}
\label{sec:system model} 
We consider a system (see Figure \ref{watersystem}) of a point-to-point communication channel where an energy harvesting transmitter $S$ sends data to a destination $D$ through a time-varying channel by judiciously adapting its transmission rate and power. The actions of $S$ are governed by three distinct exogenous processes, namely, energy harvesting, packet arrival and channel fading. We consider the system in discrete time and over a finite horizon divided by equal time slots. Let $\lbrace H_n \rbrace$, $\lbrace B_n \rbrace$ and $\lbrace \gamma_n \rbrace$ be discrete time sequences over the finite horizon $n=1,2,\ldots N$, representing energy arrivals, packet arrivals and channel gain, respectively, over a transmission window of $N<\infty$ slots, where $n$ is the time slot index. Particularly, $H_n$ is the amount of energy that becomes available in slot $n$ (harvested during slot $n-1$), $B_n$ is the amount of data that becomes available at the beginning of slot $n$ and $\gamma_n$ is the channel gain observed at slot $n$. 

Let $e_{n}$ and $b_{n}$ be energy and data buffer levels at slot $n$, where transmit power $\rho_n$ is used in slot $n$ and the received power is $\rho_{n}\gamma_n$. 

The transmit power and rate decisions $\rho_{n}$ and $r_{n}$ are assumed to obey a one-to-one relation $r_{n}= f(1+\rho_{n}\gamma_{n})$\footnote{The function $f(\cdot)$ is  a general performance function as in ~\cite{5464730}.}  , where the function $f(\cdot)$ has the following properties: 
\begin{itemize}
\item $f(x)$ is concave, increasing and differentiable.
\item $f(1)=0$ , $f'(1+x)< \infty$ and $\displaystyle\lim_{x\rightarrow \infty}f'(1+x)=0$.
\end{itemize}
The update equations for energy and data buffers can be expressed as in below:

Update Equation for the Energy Buffer:
\begin{equation}
 e_{n+1} =  e_{n} + H_{n} - \rho_{n},\rho_{n} \leq e_{n} \mbox{, for all $n$.}
\label{eq:energybufferww}
\end{equation}

Update Equation for the Data Buffer:
\begin{equation}
b_{n+1}= b_{n} + B_{n} - f(1+\rho_{n}\gamma_{n}), f(1+\rho_{n}\gamma_{n})\leq b_{n} \mbox{, for all $n$.}
\label{eq:databufferww}
\end{equation}

\begin{figure}[htpb]
    \centering \includegraphics[scale=0.68]{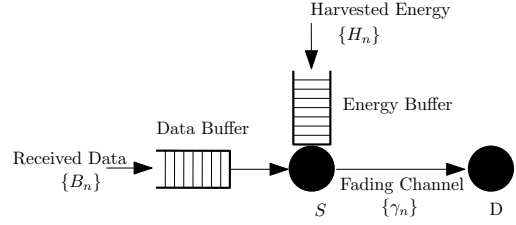}

\caption{An illustration of the system model.}
\label{watersystem} 
\end{figure}
\section{Offline Problem}
\label{sec:offlinesolution}

We consider the following offline problem over a finite horizon of $N$ slots:
\[
\mbox{Maximize}  \displaystyle\sum_{l=1}^{N} f(1+\rho_{l}\gamma_{l})   
\]
\[
\mbox{subject to constraints in  (\ref{eq:energybufferww}) and (\ref{eq:databufferww})} 
\] 
As the problem is offline, we assume $\lbrace H_n \rbrace$, $\lbrace B_n \rbrace$ and $\lbrace \gamma_n \rbrace$ time sequences are a priorly known. Accordingly, energy and data constraints can be completely determined as in the following inequalities:
\begin{equation}
\displaystyle\sum_{l=n}^{n+u}\rho_{l} \leq e_{n}+ \displaystyle\sum_{l=n+1}^{n+u} H_{l}, u=1,2,....,(N-n),
\label{eq:ec}
\end{equation}
\[
\rho_{n} \leq e_{n} \mbox{, for all $n$.}
\]
\begin{equation}
\displaystyle\sum_{l=n}^{n+v} f(1+\rho_{l}\gamma_{l})\leq b_{n}+ \displaystyle\sum_{l=n+1}^{n+v} B_{l}, v=1,2,....,(N-n)
\label{eq:bc}
\end{equation}
\[
f(1+\rho_{n}\gamma_{n})\leq b_{n} \mbox{, for all $n$.}
\]
We make the following definitions to characterize offline policies and depict a clear distinction between the concepts of energy efficiency and throughput maximization.
\begin{definition}
Any  collection of power level decisions $\mathbf{\rho}=(\rho_{1},\rho_{2},....,\rho_{N})$, satisfying energy and data constraints in (\ref{eq:ec}) and (\ref{eq:bc}), is a {\emph{feasible offline schedule}}.
\end{definition}
\begin{figure}[htpb]
    \centering \includegraphics[scale=0.68]{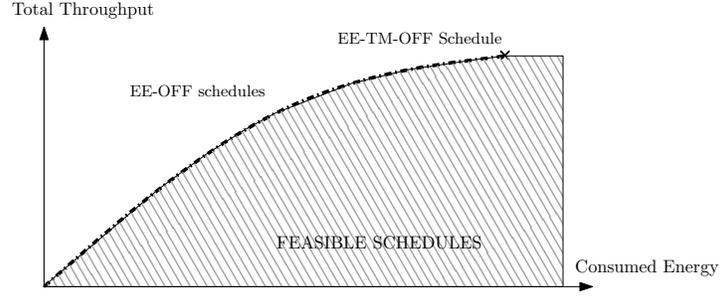}

\caption{An illustration of feasible offline schedules in terms of achieved total throughput versus consumed energy.}
\label{schedulesvw} 
\end{figure}
\begin{definition}
An {\emph{energy efficient offline transmission (EE-OFF) }} schedule is a feasible offline schedule such that there is no other feasible offline schedule that can achieve higher throughput  by consuming the same total amount of energy or achieve the same throughput by consuming less energy for a given realization of $\lbrace H_n \rbrace$, $\lbrace B_n \rbrace$ and $\lbrace \gamma_n \rbrace$ time series. 
\end{definition}
\begin{definition}
Among all EE-OFF schedules, those that achieve the maximum throughput \footnote{Note that not all feasible offline schedules that maximize the total throughput are EE-TM-OFF schedules. A schedule can be throughput optimal by delivering the data received during transmission but this can be done by consuming more energy than the corresponding EE-TM-OFF schedule.}  are called {\emph{energy efficient thoughput maximizing offline transmission (EE-TM-OFF)}} schedules.
\end{definition}
To identify the schedules in an alternative way, we define \textit{water levels} which will be useful in  Theorem \ref{optimalwater}.
\begin{definition}
A water level $w_{n}$  is the unique solution of the following:
\[
\rho_{n} = \frac{1}{\gamma_{n}}\left[ (f')^{-1}(\frac{1}{w_{n}\gamma_{n}})-1\right]^{+} 
\]
\end{definition}
\begin{prop}
The water level $w_n$ is non-decreasing in $\rho_{n}$ and  $ f(1+\rho_{l}\gamma_{l})$.
\end{prop}
\begin{proof}
As $f(\cdot)$ is increasing and concave,  $(f')^{-1}(\frac{1}{w_{n}\gamma_{n}})$ is non-decreasing in $w_n$.
\end{proof}
\begin{remark}
\label{derivativeremark}
For $\rho_{n}> 0$, the partial derivative of  $f(1+\rho_{n}\gamma_{n})$ with respect to $\rho_{n}$ is equal to $\frac{1}{w_{n}}$. 
\end{remark}
Clearly, any power level $\rho_{n}$ can be obtained from a properly chosen water level $w_{n}$. Hence, any offline transmission schedule can be also defined by corresponding water levels $(w_{1},w_{2},....,w_{n})$. 

For the solution of offline throughput maximization problem, it will be shown in Theorem \ref{optimalwater} that offline optimal water level for an EE-TM-OFF schedule is the maximum water level that barely empties data or energy buffer if it is applied continuosly.
\begin{theorem}
\label{optimalwater}
In an EE-OFF scheme, the water level $w_{n}$ is bounded as:
\[
w_{n}\leq \min \lbrace w_{n}^{(e)}(w_{n}), w_{n}^{(b)}(w_{n}) \rbrace
\]
where
\[
w_{n}^{(e)}(w_{n})= \displaystyle\min_{u=0,...,(N-n)} \frac{e_{n}+ \displaystyle\sum_{l=n+1}^{n+u} H_{l} + \displaystyle\sum_{l=n}^{n+u} K_{l}^{(e)}(w_{n}) }{u+1} 
\]
\[
w_{n}^{(b)}(w_{n}) =\displaystyle\min_{v=0,...,(N-n)} \frac{b_{n}+ \displaystyle\sum_{l=n+1}^{n+v} B_{l} + \displaystyle\sum_{l=n}^{n+v} K_{l}^{(b)}(w_{n})  }{v+1}
\]
\[
K_{l}^{(e)}(w_{n}) = w_{n}- \frac{1}{\gamma_{l}}\left[ (f')^{-1}(\frac{1}{w_{n}\gamma_{l}})-1\right]^{+}, K_{l}^{(b)}(w_{n}) = w_{n} - f\left( 1+\left[ (f')^{-1}(\frac{1}{w_{n} \gamma_{l}})-1\right]^{+} \right)  
\]

Particularly, water levels in an EE-TM-OFF schedule should satisfy the inequality above with equality, i.e. $w_{n}^{*}=\min \lbrace w_{n}^{(e)}(w_{n}^{*}), w_{n}^{(b)}(w_{n}^{*}) \rbrace$ for all $n$ in $\lbrace 1,2, ..... , N\rbrace$.
\end{theorem} 
\begin{proof}
See the Appendix.
\end{proof} 
\section{Offline Problem with Logarithmic Rate Function}
\label{sec:problemdefinition} 
In the offline problem, the throughput function $f(\cdot)$ could be chosen as $\frac{1}{2}\log_{2}(\cdot)$ that represents the AWGN capacity of the channel.  Then, the water level $w_{n}$ in this case determines the power level $\rho_{n}$ as $\rho_{n} = \frac{1}{\gamma_{n}}\left[ \frac{\ln(2)}{2}w_{n}\gamma_{n}-1\right]^{+}$.
For this case, an EE-TM-OFF schedule can be obtained by setting the water level $w_{n}$\footnote{Since $\frac{\ln(2)}{2}$ is a constant, in the rest, we will reset $w_{n}$ to $\frac{\ln(2)}{2}w_{n}$ in order to simplify the notation.}  to $\min\lbrace w_{n}^e, w_{n}^b\rbrace$ for each time slot $n$ where $w_{n}^{e}$ and $w_{n}^{b}$ are defined as follows:
\begin{footnotesize}
\begin{equation}
w_{n}^{e} =\displaystyle\min_{u=0,...,(N-n)} \frac{e_{n}+ \displaystyle\sum_{l=n+1}^{n+u} H_{l}+\displaystyle\sum_{l=n}^{n+u}  M_{l}^{(e)}(w_{n}) }{u+1}
\label{eq:eqwe}   
\end{equation}
\end{footnotesize}
\begin{footnotesize}
\begin{equation}
\log_{2}(w_{n}^{b})=\displaystyle\min_{v=0,...,(N-n)} \frac{b_{n}+ \!\!\!\!\!\displaystyle\sum_{l=n+1}^{n+v} \!\!\!\!\!\ B_{l} +\frac{1}{2}\displaystyle\sum_{l=n}^{n+v}  M_{l}^{(b)}(w_{n}) }{\frac{1}{2}(v+1)}   
\label{eq:eqwb} 
\end{equation}
\end{footnotesize}
where 
\[
 M_{l}^{(e)}(w_{n})= \min \left\lbrace \frac{1}{\gamma_{l}},w_{n}\right\rbrace , M_{l}^{(b)}(w_{n})= \log_{2}\left( \min  \left\lbrace \frac{1}{\gamma_{l}},w_{n}\right\rbrace\right)
\]
\label{eetmoffwater}
The characterization of the offline optimal water level can be explicitly expressed as in the above. Due to the correction terms $M_{l}^{(e)}(w_{n})$ and $M_{l}^{(b)}(w_{n})$, the offline optimal water level $w_{n}^{*}$ corresponds to the unique fixed point of $\min\lbrace w_{n}^e, w_{n}^b\rbrace$ and should be computed iteratively. To find the water level satisfying $\min\lbrace w_{n}^e, w_{n}^b\rbrace$, any  fixed point iteration method can be used. For example, the throughput maximizing water level $w_{n}^{*}$ can be found by iteratively evaluating $\min\lbrace w_{n}^e, w_{n}^b\rbrace$  as follows:
\begin{equation}
\label{eetmoffwateritarate}
w_{n}^{(k+1)}=\vert_{w_{n}=w_{n}^{(k)}}\min\lbrace w_{n}^e, w_{n}^b\rbrace
\end{equation}
where $w_{n}^{(1)}=w^{max}_{n}$ which is  guaranteed to be higher than  $w_{n}^{*}$.
The proposotion in the below states that the iteration in Eq. \ref{eetmoffwateritarate} converges.
\begin{prop}
\label{convergenceprop}
The sequence of water level iterations, $w_{n}^{(1)}, w_{n}^{(2)},....$ converges to $w_{n}^{*}$.
\end{prop}
\begin{proof}
From (\ref{eq:eqwe}),(\ref{eq:eqwb}) and (\ref{eetmoffwateritarate}), $w_{n}^{(k+1)}$ is decreasing with  decreasing $w_{n}^{(k)}$. Accordingly, if $w_{n}^{(k+1)}<w_{n}^{(k)}$ for some $k$, then $w_{n}^{(k+2)}<w_{n}^{(k+1)}$ should be true and setting $w_{n}^{(1)}$  to a large enough value can guarantee that $w_{n}^{(2)}<w_{n}^{(1)}$ . As $w_{n}^{(k)}$'s are bounded below by zero, the iterations converge. Unless $w_{n}^{*}$ is reached, the iterations have not stopped, hence the iterations will converge to $w_{n}^{*}$ if $w_{n}^{(1)}$ is above $w_{n}^{*}$.
\end{proof}
The offline optimal power level $\rho_{n}^{*}$ that maximizes total throughput can be approached by computing the sequence, $w_{n}^{(1)},w_{n}^{(2)},....$, which converges  $w_{n}^{*}$ by Proposition \ref{convergenceprop}.
\begin{equation}
\rho_{n}^{*}=\lim_{k \rightarrow \infty } \left[w_{n}^{(k)}-\frac{1}{\gamma_{n}}\right]^{+}
\end{equation}
\vspace{-0.099in} 
\begin{figure}[htpb]
    \centering \includegraphics[scale=0.44]{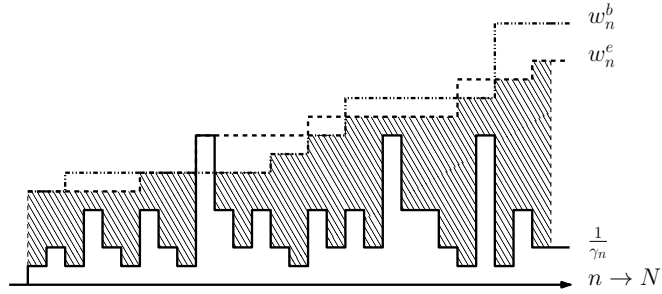}

\caption{An illustration of an EE-TM-OFF policy.}
\label{offlinecavew} 
\end{figure}
\section{Online Problem}
The online problem formulation is an online counterpart of the offline problem with logarithmic \footnote{For the sake of simplicity, the logarithmic rate function will be used in online formulations. However, similar formulations and results can be obtained also for the general concave function $f(\cdot)$.}  rate function. We formulate the problem as a dynamic programming to  maximize the expected total throughput.   
\begin{figure}[htpb]
\centering
  \begin{psfrags}
    \psfrag{A}[t]{Slot Index}
    \psfrag{B}[b]{Water Levels}
    \psfrag{C}[l]{\tiny{EE-TM-OFF}}
    \psfrag{D}[l]{\tiny{Expected EE-TM-OFF}}
    \includegraphics[scale=0.27]{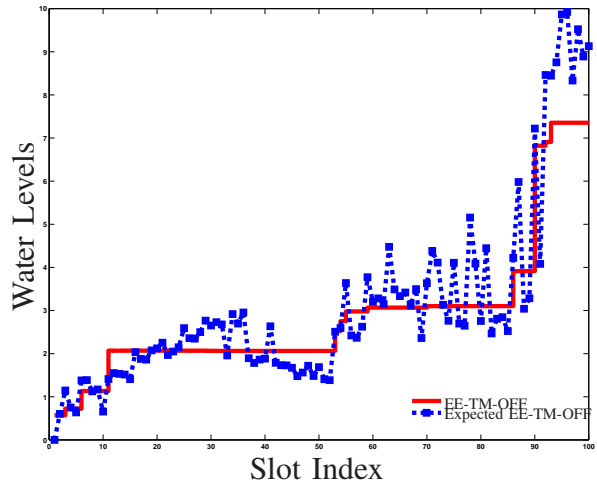}
    \end{psfrags}
\caption{Sample water levels of an EE-TM-OFF schedule and the expected EE-TM-OFF given knowledge for a sample realization where energy harvesting, data arrival and channel fading processes are generated by 4 state DTMCs.}
\label{fig:waterlevelsw}
\end{figure}
Let $x_{n}=(e_{n},b_{n},\gamma_{n})$ be the state vector, $\theta_{n}=(H_{1}^{n},B_{1}^{n},\gamma_{1}^{n})$ be the history and $X_{n}=(H_{n+1},B_{n+1},\gamma_{n+1}-\gamma_{n})$ exogeneous processes at the slot $n$.

Define  $A(x_{n})$ as the set of admissible decisions such that  $[w_{n}-\frac{1}{\gamma_{n}}]^{+} \leq e_{n}$ and $[\log_{2}(w_{n}\gamma_{n})]^{+}  \leq b_{n}$, $\forall w_{n} \in A(x_{n})$. For $w_{n} \in A(x_{n})$, the dynamic program for throughput maximization  can be written as below:
\begin{equation}
\hat{V}_{n \mid \theta_{n}}^{*}(x_{n})=\displaystyle\max_{w_{n} \in A(x_{n})} \hat{V}_{n \mid \theta_{n}}(w_{n}, x_{n})
\label{eq:naivedynamic}
\end{equation}
\begin{equation}
\hat{V}_{n \mid \theta_{n}}w_{n}, x_{n})=[\log_{2}(w_{n}\gamma_{n})]^{+} + \displaystyle E_{X_{n}} [\hat{V}_{n+1 \mid \theta_{n+1}}^{*}(x_{n}+X_{n}-\phi(w_{n};\gamma_{n}))\mid x_{n},\theta_{n}]
\end{equation}
where $\phi(w_{n};\gamma_{n})=([w_{n}-\frac{1}{\gamma_{n}}]^{+},[\log_{2}(w_{n}\gamma_{n})]^{+} ,0)$ and $\psi_{n}=(H_{n+1}^{N},B_{n+1}^{N},\gamma_{n+1}^{N})$ represents the exogeneous processes for slots between $n$ and $N$.

The solution of this dynamic programming formulation constitutes the online optimal policy maximizing expected total throughput to be achieved within the finite problem horizon. The drawback of this solution is that it suffers from the exponential time/memory computational complexity of the dynamic programming. On the other hand, when the vector $\psi_{n}=(H_{n+1}^{N},B_{n+1}^{N},\gamma_{n+1}^{N})$ is deterministic, the online problem is no different than the offline problem. The solution to the offline problem for the realization of $\psi_{n}$ can be a  reference for the online problem.  We observed that a policy, which simply applies the statistical average of EE-TM-OFF water levels as its online water level at each and every time slot, typically closely  follows the original EE-TM-OFF schedule (Fig. \ref{fig:waterlevelsw}). Motivated by this observation, we consider EE-TM-OFF  decisions as stochastic processes in the online problem domain.  The next subsection will introduce an alternative dynamic programming formulation of  minimizing the expected throughput loss of the online decisions with respect to the corresponding offline optimal decisions.
\subsection{Online Solution Based On Offline Solution} 
 Let $\tilde{w}_{n}^{*}=\tilde{w}_{n}^{*}(x_{n})$ be the offline optimal water level which is a random variable generated over the realizations of $\psi_{n}$ given the state vector $x_{n}$. Then, the total throughput achieved by applying offline optimal water levels until the end of transmission time window can be expressed as:
\begin{equation}
\tilde{V}_{n \mid \theta_{n}}^{*}(x_{n})=[\log_{2}(\tilde{w}_{n}^{*}\gamma_{n})]^{+}+\tilde{V}_{n+1 \mid \theta_{n+1}}^{*}(x_{n}+X_{n}-\phi(\tilde{w}_{n}^{*};\gamma_{n}))
\end{equation}
The online throughput maximization problem can be reformulated by the following cost minimization problem:
\begin{equation}
J_{n \mid \theta_{n}}^{*}(x_{n})=\displaystyle\min_{w_{n} \in A(x_{n})} J_{n \mid \theta_{n}}(w_{n},x_{n})
\end{equation}
where
\begin{equation}
J_{n \mid \theta_{n}}(w_{n},x_{n})=E_{\psi_{n}}[\tilde{V}_{n \mid \theta_{n}}^{*}(x_{n})\mid x_{n},\theta_{n}]-\hat{V}_{n \mid \theta_{n}}(w_{n}, x_{n})
\end{equation}
The cost function $J_{n \mid \theta_{n}}(w_{n},x_{n})$ can be separated into two parts: 
\begin{itemize}
\item The expected  throughput achieved by applying offline water levels for slots $[n,N]$ minus the expected throughput achieved by applying the decision $w_{n}$ at the slot $n$, then applying offline optimal water levels for the rest, i.e. in $[n+1,N]$. Let $E_{\psi_{n}}[\tilde{F}_{n}(\tilde{w}_{n}^{*}, w_{n})\mid x_{n},\theta_{n}]$ represent this term.
\item The expected throughput achieved by applying offline water levels for slots $[n+1,N]$ minus the expected total throughput achieved by online optimal decision for slots $[n+1,N]$ after the decision $w_{n}$ is applied at the slot $n$. Let $D_{n+1 \mid \theta_{n+1}}^{*}(x_{n}+X_{n}-\phi(w_{n};\gamma_{n}))$ represent this term.
\end{itemize}
\begin{equation}
J_{n \mid \theta_{n}}(w_{n},x_{n}) ) = E_{\psi_{n}}[\tilde{F}_{n}(\tilde{w}_{n}^{*}, w_{n})\mid x_{n},\theta_{n}]+D_{n+1 \mid \theta_{n+1}}^{*}(x_{n}+X_{n}-\phi(w_{n};\gamma_{n}))
\label{offonline} 
\end{equation}
%
Clearly, both of the terms are nonnegative for any $w_{n}$ since, by definition, EE-TM-OFF schedules are superior to online throughput maximizing schedules for any given realization.   

The first term $E_{\psi_{n}}[\tilde{F}_{n}(\tilde{w}_{n}^{*}(x_{n}), w_{n})\mid x_{n},\theta_{n}]$ is the conditional expectation of the variable $\tilde{F}_{n}(\tilde{w}_{n}^{*}, w_{n})$ as  follows:
\[
\tilde{F}_{n}(\tilde{w}_{n}^{*}, w_{n})= ( \log_{2}(\tilde{w}_{n}^{*}\gamma_{n}) )^{+}-(\log_{2}(w_{n}\gamma_{n}))^{+}
\]
\[
+\tilde{V}_{n+1 \mid \theta_{n+1}}^{*}(x_{n}+X_{n}-\phi(\tilde{w}_{n}^{*};\gamma_{n}))-\tilde{V}_{n+1 \mid \theta_{n+1}}^{*}(x_{n}+X_{n}-\phi(w_{n};\gamma_{n}))
\]
The equation in (\ref{offonline}) can be rewritten as in below:
\begin{equation}
J_{n \mid \theta_{n}}(w_{n},x_{n})=E_{\psi_{n}}[F_{n}(\tilde{w}_{n}^{*}, w_{n})\mid x_{n},\theta_{n}]+D_{n+1 \mid \theta_{n+1}}^{*}(x_{n}+X_{n}-\phi(w_{n};\gamma_{n}))
\label{offonlineF} 
\end{equation}
where $F_{n}(\tilde{w}_{n}^{*}, w_{n})=E_{\psi_{n}}[\tilde{F}_{n}(\tilde{w}_{n}^{*}, w_{n})\mid \tilde{w}_{n}^{*},x_{n},\theta_{n}]$.

Accordingly, the function $F_{n}(\tilde{w}_{n}^{*}, w_{n})$ can be seen as a loss function for the decision $w_{n}$ since it corrresponds to the throughput loss that cannot be recovered  even with offline optimal policies. The expectation of this loss term will be called as the \emph{immediate loss} of the decision $w_{n}$ as we define in below.
\begin{definition}
Define $E_{\psi_{n}}[F_{n}(\tilde{w}_{n}^{*}, w_{n})\mid x_{n},\theta_{n}]$ as the \emph{immediate loss} of the decision $w_{n}$. 
\end{definition}  
On the other hand, the second term $D_{n+1 \mid \theta_{n+1}}^{*}(\cdot)$ can be expressed as :
\begin{equation}
D_{n+1 \mid \theta_{n+1}}^{*}(x_{n+1})=E_{X_{n}}[J_{n+1 \mid \theta_{n+1}}^{*}(x_{n+1})\mid x_{n},\theta_{n}]
\end{equation}
where $x_{n+1}=x_{n}+X_{n}-\phi(w_{n};\gamma_{n})$.

Therefore, the problem has the following dynamic programming formulation:
\begin{equation}
J_{n \mid \theta_{n}}^{*}(x_{n})= \displaystyle\min_{w_{n} \in A(x_{n})} E_{\psi_{n}}[F_{n}(\tilde{w}_{n}^{*}, w_{n})\mid x_{n},\theta_{n}]+E_{X_{n}}[J_{n+1 \mid \theta_{n+1}}^{*}(x_{n}+X_{n}-\phi(w_{n};\gamma_{n}))\mid x_{n},\theta_{n}]
\label{offonlineDP} 
\end{equation}
As this formulation is  equivalent to the initial formulation in (\ref{eq:naivedynamic}), its solution  gives the online optimal policy. While the exact computation of this solution may also have exponential complexity, the formulation will lead us to define the \emph{immediate fill} metric which will be a vehicle toward the derivation of online solutions with performance guarantees. 
\subsection{Immediate Fill}
The performance of any online policy $w$ can be also evaluated by the ratio of its expected total throughput to the expected total throughput of the offline optimal policies.
\begin{definition}
Define the \emph{online-offline efficiency}, or simply the efficiency of an online policy $w$ as follows:
 \begin{equation}
 \eta^{w}(x_{n},\theta_{n}) = \frac{\hat{V}_{n \mid \theta_{n}}^{w}(x_{n})}{E_{\psi_{n}}[\tilde{V}_{n \mid \theta_{n}}^{*}(x_{n})\mid x_{n},\theta_{n}]}
 \end{equation}
where $\hat{V}_{n \mid \theta_{n}}^{w}(x_{n})$ is the expected total throughput achieved by the online policy $w$ given the present state $x_{n}$ and the history $\theta_{n}$.
\end{definition}
Any decision in the online schedule will incur an immediate throughput gain. However, this decision also may cause a loss of potential future throughput that would be accessible to an offline algorithm. We call this the immediate loss as we define in the previous section.
\begin{definition}
 Define the ratio of immediate gain to its sum with immediate loss as \emph{immediate fill}. For the slot $n$, let  $\mu_{n}^{w}(x_{n},\theta_{n})$ be the immediate fill of policy $w$ when the system is in $x_{n}$  with the history $\theta_{n}$. 
\begin{equation}
\mu_{n}^{w}(x_{n},\theta_{n}) = \frac{(\log_{2}(w_{n}\gamma_{n}))^{+}}{(\log_{2}(w_{n}\gamma_{n}))^{+}+E_{\psi_{n}}[F_{n}(\tilde{w}_{n}^{*}, w_{n})\mid x_{n},\theta_{n}]}
\end{equation}
\end{definition}
According to the above definition of the immediate fill $\mu_{n}^{w}(x_{n},\theta_{n})$, we will show that the minimal immediate fill of the policy $w$ lower bounds its online-offline efficiency.
\begin{figure}[htpb]
    \centering \includegraphics[scale=0.76]{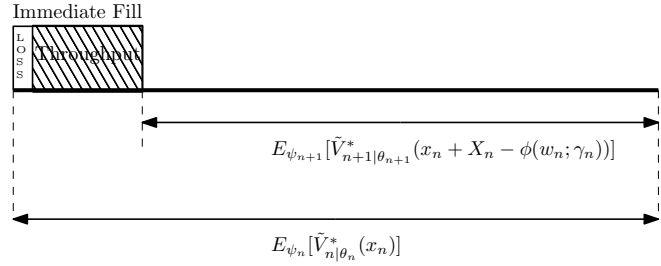}

\caption{An illustration of the immediate fill approach. The expectation of the achievable total throughput by offline optimal decisions decreases as the state of the system changes due to an online decision. Hence each online decision opens a gap between expected throughput potentials of offline optimal policy and this gap is partially filled by the throughput gain achieved within the corresponding slot.}
\label{offlinecavew} 
\end{figure}
\begin{theorem}
\label{thm:minmumfill}
The efficiency of an online policy $w$ with $w_{N}=\tilde{w}_{N}^{*}$ is lower bounded by the minimum immediate fill observed by that policy:
 \begin{equation}
 \label{lowerfill}
 \eta^{w}(x_{n},\theta_{n}) \geq \displaystyle\min_{m \geq n} \displaystyle\min_{(x_{m},\theta_{m})} \mu_{m}^{w}(x_{m},\theta_{m})
 \end{equation}  
\end{theorem} 
\begin{proof}
See the Appendix.
\end{proof} 
\begin{figure}[htpb]
\centering
  \begin{psfrags}
    \psfrag{A}[t]{$p$}
	\psfrag{B}[b]{$LB_1$}
	\psfrag{P}[b]{$LB$}  
	\psfrag{C}[l]{\scriptsize{$e_n =5$}}
	\psfrag{D}[l]{\scriptsize{$e_n =55$}}
	\psfrag{E}[l]{\scriptsize{$e_n =25$}}
	
 	\psfrag{F}[l]{\large{$N-n =25$}}
\subfloat[]{
  \includegraphics[scale=0.23]{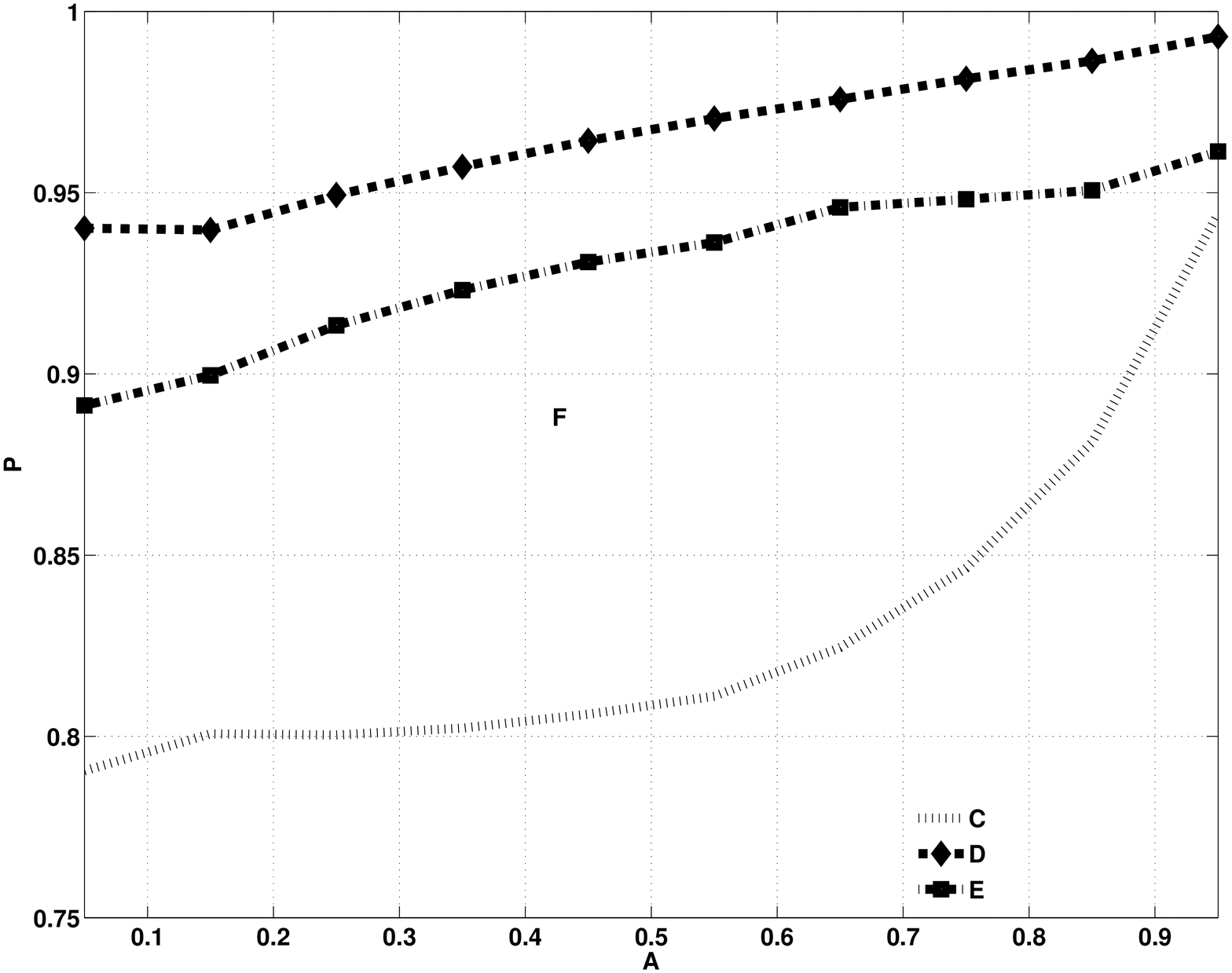}
}
\psfrag{F}[l]{\large{$N-n =5$}} 
\subfloat[]{
  \includegraphics[scale=0.23]{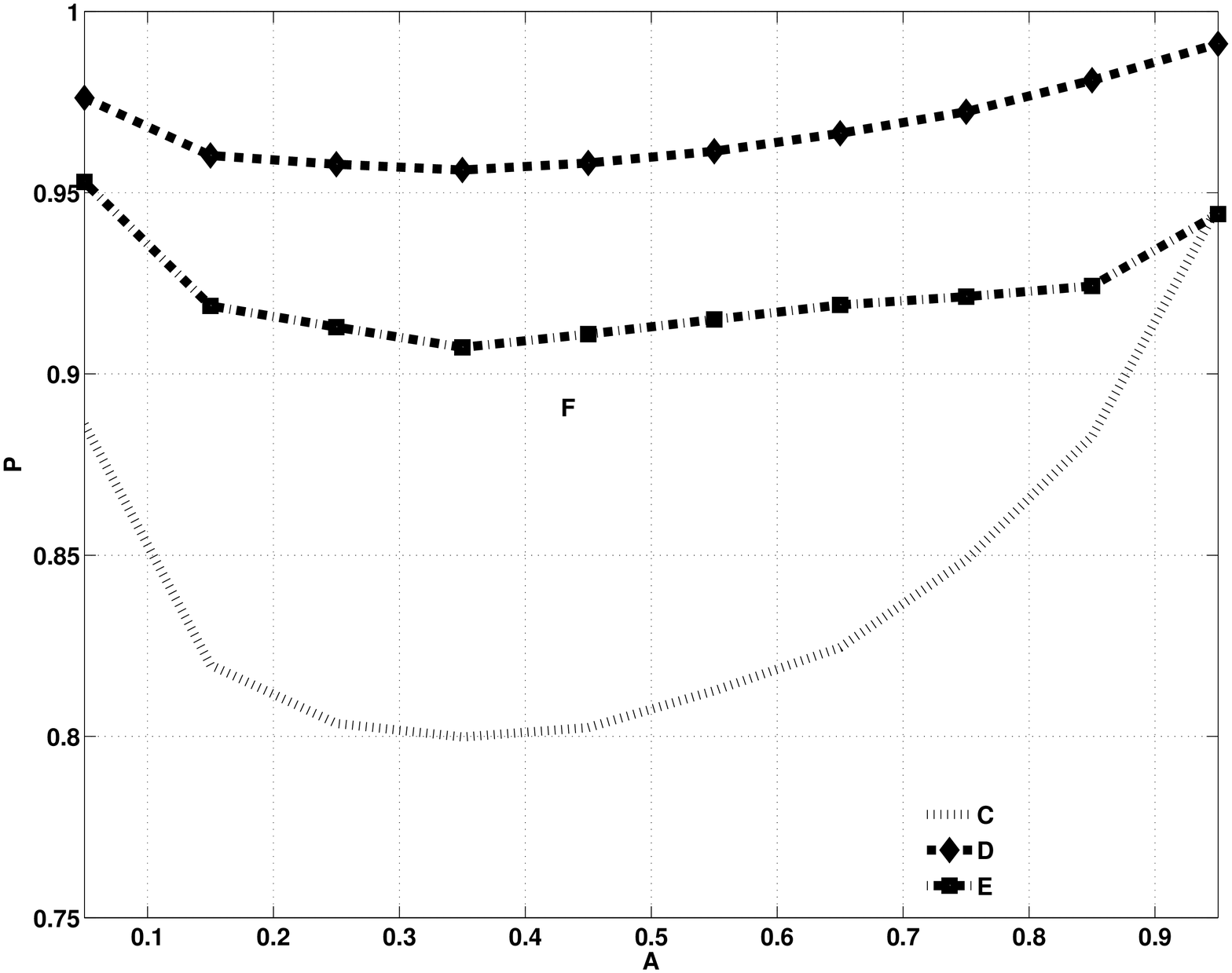}
} 	
    \end{psfrags}  
\caption{The lower bounds of the immediate fill metric  at (a) $N-n=25$ (b) $N-n=5$    for $\rho_{n}=E[\tilde{\rho}_{n}^{*}]$ policy where $\lbrace H_n \rbrace$ is a Bernoulli process with $\Pr(H_n=0)=1-p$ and $\Pr(H_n=24\text{ units})=p$.}
\label{r255} 
\end{figure}
Next section considers stochastic offline optimal decisions in a simpler case, namely static channel case, in order to demonstrate how simple bounds on immediate fill can be found and the distribution of offline optimal decisions can be characterized. 
\subsection{Results on the Static Channel Case}
In this section, we focus on the case where the channel is static, i.e. $\gamma_{n}=1$ for all $n$,  and the data buffer is always full, i.e. $b_{n}=\infty$ for all $n$. Accordingly, the online power level and the offline optimal power level can be represented by 
$\rho_{n}= w_{n}-1$ and $\tilde{\rho}_{n}^{*}= \tilde{w}_{n}^{*}-1$. Then, the offline optimal power level at slot $n$ can be expressed as:
\begin{equation}
\tilde{\rho}_{n}^{*} =\displaystyle\min_{u=0,...,(N-n)} \frac{e_{n}+ \displaystyle\sum_{l=n+1}^{n+u} H_{l}}{u+1}
\end{equation}
\begin{prop}
\label{thm:immediateprop}
Assuming that the the channel is static, i.e. $\gamma_{n}=1$ for all $n$,  and the data buffer is always full, i.e. $b_{n}=\infty$ for all $n$, the immediate fill is lower bounded as follows:
\[
\mu_{n}^{w}(x_{n},\theta_{n})\geq \frac{\ln(1+\rho_{n})}{E\left[ \ln(1+\tilde{\rho}_{n}^{*})\right] +E\left[\frac{(\rho_{n}-\tilde{\rho}_{n}^{*})^{+}}{1+\tilde{\rho}_{n+1}^{\triangleright}}\right] }
\]   
where  $\tilde{\rho}_{n+1}^{\triangleright}$ is the offline optimal decision at slot $n+1$ after the decision $\rho_{n}$ is made.
\end{prop}
\begin{proof}
See the Appendix.
\end{proof}
\begin{prop}
\label{immediatereprop} 
Let $\mu_{n}^{\check{w}}(x_{n},\theta_{n})$ represent the maximum (achievable) immediate fill at slot $n$,
i.e,  $\mu_{n}^{\check{w}}(x_{n},\theta_{n})= \displaystyle\max_{w_{n} \in A(x_{n})}\mu_{n}^{w}(x_{n},\theta_{n})$. Then, the inequality below should hold:
\[
\mu_{m}^{\check{w}}(x_{m},\theta_{m}) \geq \frac{1}{ 1 + \frac{E\left[\frac{(E[\tilde{\rho}_{n}^{*}]-\tilde{\rho}_{n}^{*})^{+}}{1+\tilde{\rho}_{n+1}^{\triangleright}}\right]}{\ln(1+E[\tilde{\rho}_{n}^{*}])}}, (LB)
\]
and it can be simplified as in the following:
\[
\mu_{m}^{\check{w}}(x_{m},\theta_{m}) \geq \frac{1}{ 1 + \frac{E\left[(E[\tilde{\rho}_{n}^{*}]-\tilde{\rho}_{n}^{*})^{+}\right]}{\ln(1+E[\tilde{\rho}_{n}^{*}])}}
\]
which implies:
\[
\mu_{m}^{\check{w}}(x_{m},\theta_{m}) \geq \frac{1}{ 1 + \frac{\sqrt{Var\left(\tilde{\rho}_{n}^{*}\right)}}{\ln(1+E[\tilde{\rho}_{n}^{*}])}},
\]
\end{prop}
\begin{proof}
See the Appendix.
\end{proof} 
In Fig. \ref{r255}, the lower bound $LB$ in Proposition \ref{immediatereprop} is plotted against varying arrival probabilities of a Bernoulli energy harvesting process at different system states of energy level $e_{n}$ and remaining number of slots $N-n$. 

Next, we consider the CDF of $\tilde{\rho}_{n}^{*}$ under Bernoulli energy harvesting assumption and characterize it for large $N$, i.e. as $N-n$ goes to infinity.
\begin{theorem}
\label{thm:ccdfrho}
Let $\lbrace H_n \rbrace$ be a Bernoulli process with $\Pr(H_n=0)=1-p$ and $\Pr(H_n=h)=p$. Then, 

(i) \[
\displaystyle\lim_{ x \rightarrow +\infty } \Pr(\tilde{\rho}_{n}^{*}< r \mid e_n = x) = 0
\]
and
\[
\displaystyle\lim_{ x \rightarrow r^{+} } \Pr(\tilde{\rho}_{n}^{*}< r \mid e_n = x) = 1
\]
(ii) for $m \in \mathbb {N}^{+}$ and $\frac{h}{m}  < e_{n}$,
\[
\displaystyle\lim_{ N \rightarrow +\infty } \Pr(\tilde{\rho}_{n}^{*}< \frac{h}{m} \mid e_n = x) = \Phi(m)^{\lfloor\frac{mx}{h}\rfloor} 
\]
where $\Phi(m)$ function is the minimum value in $(0,1]$ satisfying the following equation:
\[
p\Phi(m)^{m}-\Phi(m)+1-p=0
\]
\end{theorem} 
\begin{proof}
See the Appendix.
\end{proof} 
\begin{figure}[htpb]
\centering
  \begin{psfrags}
    \psfrag{A}[b]{$r$}
	\psfrag{B}[t]{$\Pr(\tilde{\rho}_{n}^{*}< r \mid e_{n} = x)$}
	\psfrag{C}[l]{\scriptsize{Simulated CDF}}
	\psfrag{D}[l]{\scriptsize{Computed CDF}}
    \centering \includegraphics[scale=0.26]{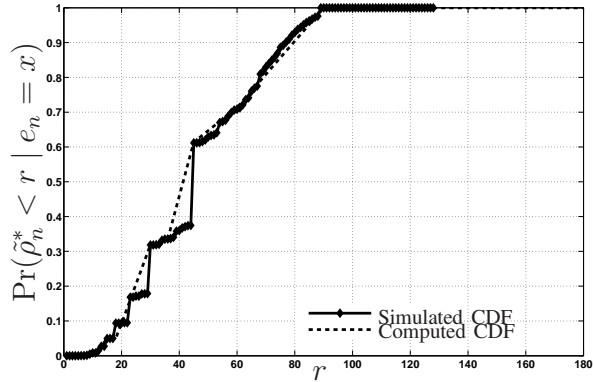}
    \end{psfrags}  
\caption{A comparison of Monte Carlo simulated CDF of $\tilde{\rho}_{n}^{*}$ at $N-n=99$ versus the CDF  of $\tilde{\rho}_{n}^{*}$ computed for $N-n \rightarrow +\infty$ using the result in Theorem \ref{thm:ccdfrho}  where $e_{n}=88$ and $\lbrace H_n \rbrace$ is a Bernoulli process with $\Pr(H_n=0)=0.55$ and $\Pr(H_n=180\text{ units})=0.45$.}
\vspace{-0.099in}
\label{sample12} 
\end{figure}
\subsection{Online Heuristic}
\label{sec:onlinepolicies}
The online problem formulation in the previous sections assumes statistical information on exogeneous processes energy harvesting, packet arrival and channel fading. Then, the offline optimal decisions take these processes  as their inputs in Eq. (\ref{eq:eqwe}) and Eq. (\ref{eq:eqwb}). 

Alternatively, a heuristic policy could use Eq. (\ref{eq:eqwe}) and Eq. (\ref{eq:eqwb}) with estimated values of $\displaystyle\sum_{l=n+1}^{n+u} H_{l}$, $\displaystyle\sum_{l=n+1}^{n+u} B_{l}$, $\displaystyle\sum_{l=n}^{n+u}  M_{l}^{(e)}(w_{n})$ and $\displaystyle\sum_{l=n}^{n+u}  M_{l}^{(b)}(w_{n})$. We propose such a policy where $\displaystyle\sum_{l=n+1}^{n+u} H_{l}$, $\displaystyle\sum_{l=n+1}^{n+u} B_{l}$, $\displaystyle\sum_{l=n}^{n+u}  M_{l}^{(e)}(w_{n})$ and $\displaystyle\sum_{l=n}^{n+u}  M_{l}^{(b)}(w_{n})$ are estimated through observed time averages giving  the estimated values of $w_{n}^{e}$ and $w_{n}^{b}$ as follows:
\begin{equation}
\hat{w}_{n}^{e}= \left\{ \begin{array}{ll}
       \frac{e_{n}-\bar{H}_{n}}{(N-n)}+\bar{H}_{n}+\bar{M}_{n}^{(e)}(w_{n})  & \mbox{; $e_{n} \geq \bar{H}_{n}$}\\
       e_{n}+\bar{M}_{n}^{(e)}(w_{n})& \mbox{; o.w.  }\\
        \end{array} \right. 
\end{equation}
\begin{equation}
\log_{2}(\hat{w}_{n}^{b})= \left\{ \begin{array}{ll}
       \frac{2(b_{n}-\bar{B}_{n})}{(N-n)}+\bar{B}_{n}+\bar{M}_{n}^{(b)}(w_{n})  & \mbox{; $b_{n} \geq \bar{B}_{n}$}\\
       2b_{n}+\bar{M}_{n}^{(b)}(w_{n})& \mbox{; o.w.  }\\
        \end{array} \right. 
\end{equation}
where
\vspace{-0.099in}
\[
\bar{H}_{n}=\frac{1}{n}\displaystyle\sum_{l=1}^{n} H_{l}, \bar{B}_{n}=\frac{1}{n}\displaystyle\sum_{l=1}^{n} B_{l}
\]
\vspace{-0.099in}
\[
\bar{M}_{n}^{(e)}(w_{n})=\frac{1}{n}\displaystyle\sum_{l=1}^{n} M_{l}^{(e)}(w_{n}),\bar{M}_{n}^{(b)}(w_{n})=\frac{1}{n}\displaystyle\sum_{l=1}^{n} M_{l}^{(b)}(w_{n})
\]
The estimate of the throughput maximizing water level can be computed iteratively:
\[
\hat{w}_{n}^{(k+1)}=\vert_{w_{n}=\hat{w}_{n}^{(k)}}\min \left\lbrace \hat{w}_{n}^{e},\hat{w}_{n}^{b}\right\rbrace 
\]
where $\hat{w}_{n}^{(k)}$ is the $k$th iteration of the estimated value of throughput maximizing water level and $\hat{w}_{n}^{(1)}=\min \left\lbrace e_{n},2^{2b_{n}}\right\rbrace $.
\section{Numerical Study of the Online vs Offline Policies}

The purpose of the numerical study is to compare the online heuristic with the offline optimal policy, under Markovian arrival processes. For the packet arrival process, a Markov model having two states as no packet arrival state and a packet arrival of constant size 10 KB per slot state with transition probabilities $q_{00}=0.9$, $q_{01}=0.1$, $q_{10}=0.58$, $q_{11}=0.42$ where slot duration is $1$ms and the transmission window is $N=100$ slots. Gilbert-Elliot channel is assumed where good ($\gamma^{good}=30$) and bad ($\gamma^{bad}=12$) states appear with equal probabilities, i.e. $P(\gamma_{n}=\gamma^{good})=P(\gamma_{n}=\gamma^{bad})=0.5$. Similarly, in energy harvesting process, energy harvests of $50$nJs are assumed to occur with a probability of $0.5$ at each slot.

For a typical sample realization of packet arrival, energy harvesting and channel fading processes, water level profiles of throughput maximizing offline optimal policy and online heuristic policy are shown in Fig. \ref{sample12} (a) and (b). Fig. \ref{sample12} (a) shows water level profiles when transmission window size
$N$ is set to $100$ slots and Fig. \ref{sample12} (b) shows water level profiles when transmission window size is extended to $200$ slots. In the first $100$ slot, water level profiles are similar to each other though, due to the relaxation of the deadline constraint, both optimal and heuristic water levels sligthly decrease when transmission window size is doubled. 

To illustrate the effect of transmission window size, average throughput performances and energy consumption of throughput maximizing offline optimal policy and online heuristic are compared against varying transmission window size in Fig. \ref{avwthroughput} (a) and (b), respectively. The average performances of both offline optimal policy and online heuristic tend to saturate as transmission window size increases beyond $100$ slots. The experiment is repeated in Fig. \ref{avwthroughputmemory}, for the case where energy harvesting process has a memory remaning in the same state with $0.9$ probability and switching to other state with probability $0.1$.

 In Fig. \ref{fig:maxthroughputcomppowerhalv}, our online heuristic is compared with the ``Power-Halving" policy proposed in ~\cite{6202352}.The power-halving policy basically operates as follows: in each slot except the last one, it keeps half the stored energy in the battery, and uses the other half. It has been shown in ~\cite{6202352}, the average throughput performance of the ``Power-Halving" policy can reach  $\%80-\%90$ of average throughput of offline optimal policy. On the other hand, our online heuristic proposed in this paper uses casual information on energy-data arrivals and channels states to achieve average throughput rate much closer to offline optimal average throughput rates.

Note that the parameters of Markov processes have been arbitrarily chosen as it is hard to cover a wide range of possible settings. 
\begin{figure}[htpb]
\centering
  \begin{psfrags}
    \psfrag{A}[b]{Power (in mW)}
	\psfrag{B}[t]{Window Size (in number of slots)}
	\psfrag{C}[l]{\scriptsize{Offline Optimal}}
	\psfrag{D}[l]{\scriptsize{Online Heuristic}}
\subfloat[]{
  \includegraphics[scale=0.26]{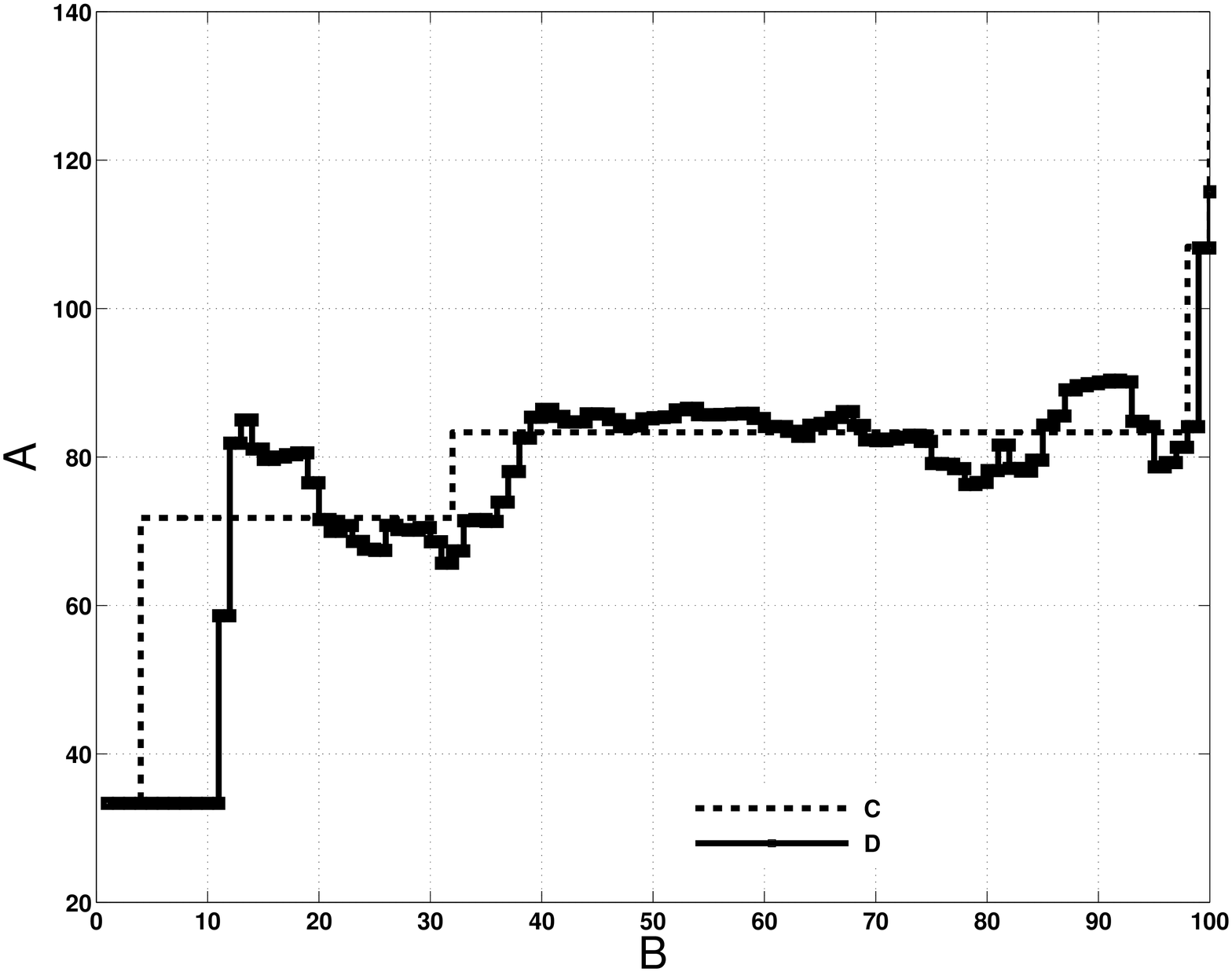}
}
\subfloat[]{
  \includegraphics[scale=0.26]{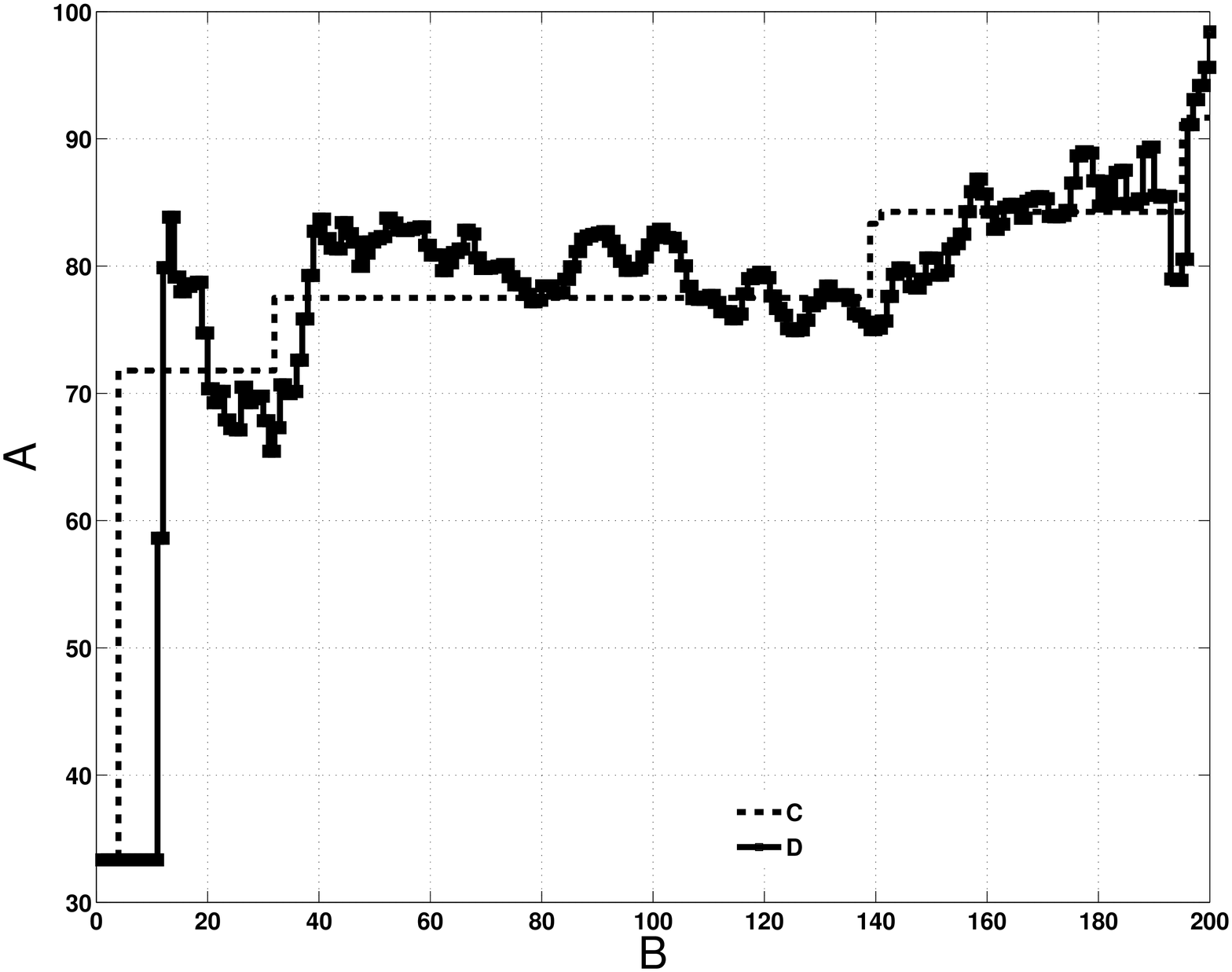}
}
    \end{psfrags}  
\caption{Water level profiles of throughput maximizing offline optimal policy  and online heuristic policy  for a sample realization of packet arrival, energy harvesting and channel fading processes when $N=100$ (a) and $N=200$ (b).}
\vspace{-0.099in}
\label{sample12} 
\end{figure}
\begin{figure}[htpb]
\centering
  \begin{psfrags}
    \psfrag{A}[t]{Window Size (in number of slots)}
	\psfrag{B}[b]{Avg. Throughput (in Mbit/s)}
	\psfrag{P}[b]{Energy Consumption per slot (in nJ)}  
	\psfrag{C}[l]{\scriptsize{Offline Optimal}}
	\psfrag{D}[l]{\scriptsize{Online Heuristic}}
	\subfloat[]{
  \includegraphics[scale=0.26]{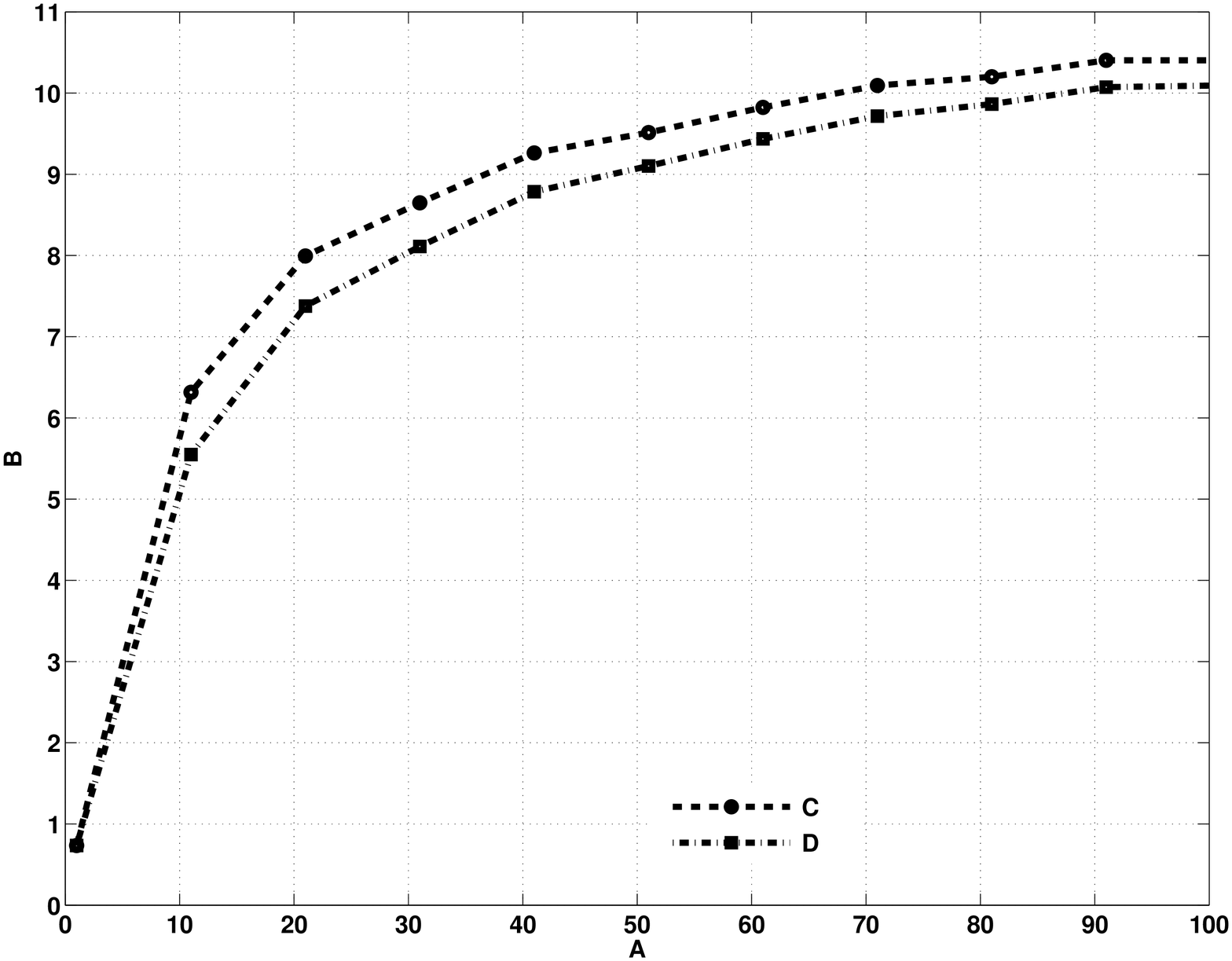}
	}
\subfloat[]{
  \includegraphics[scale=0.26]{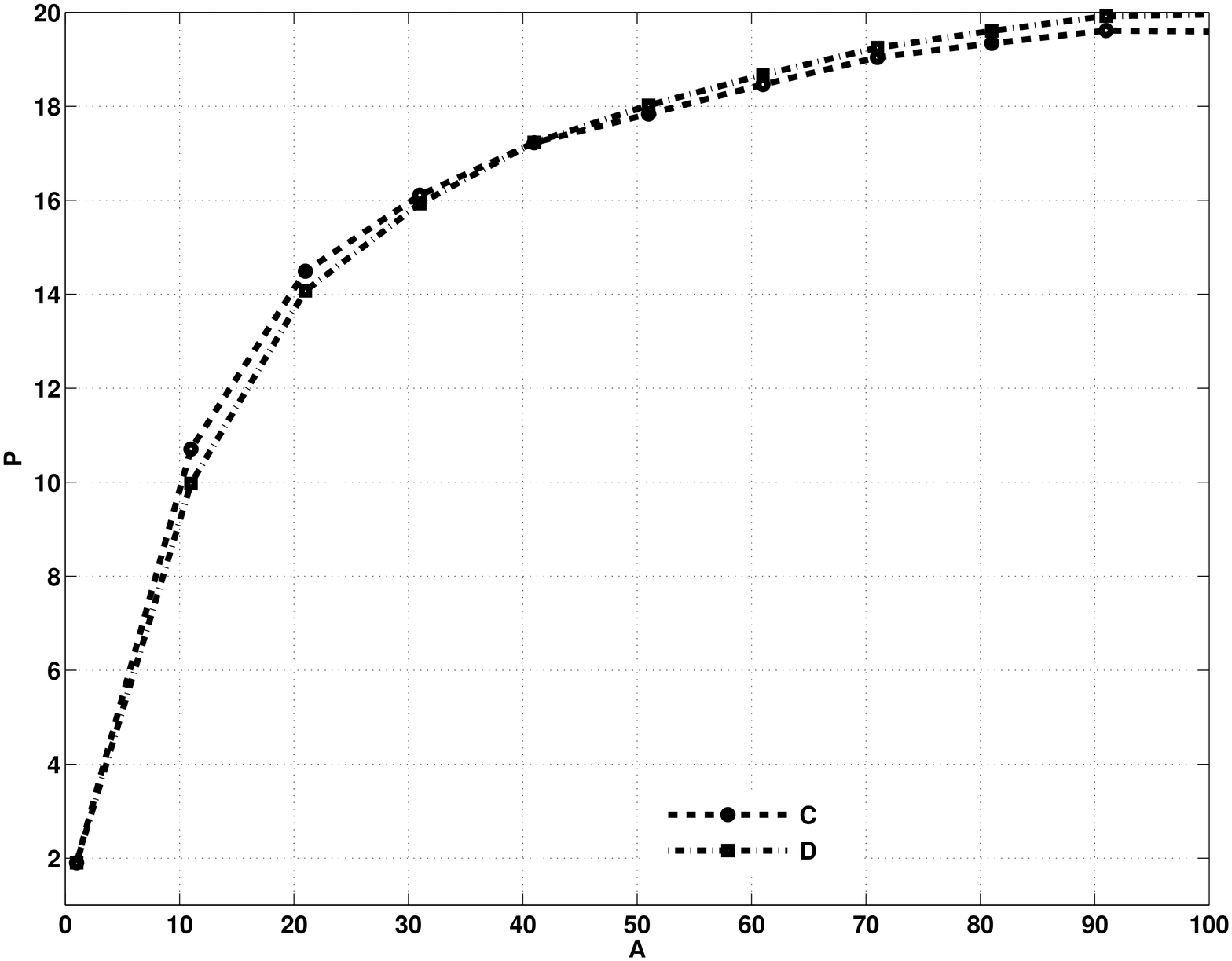}
	}
%
%
    \end{psfrags}  
\caption{Average throughput (a) and energy consumption per slot (b) comparison of throughput maximizing offline optimal policy and online heuristic policy  against varying transmission window size for stationary energy harvesting.}
\vspace{-0.099in}
\label{avwthroughput} 
\end{figure}
\vspace{-0.099in}
\vspace{-0.099in}  
\begin{figure}[htpb]
\centering
  \begin{psfrags}
    \psfrag{A}[t]{Window Size (in number of slots)}
	\psfrag{B}[b]{Avg. Throughput (in Mbit/s)}
	\psfrag{P}[b]{Energy Consumption per slot (in nJ)}  
	\psfrag{C}[l]{\scriptsize{Offline Optimal}}
	\psfrag{D}[l]{\scriptsize{Online Heuristic}}
	\subfloat[]{
  \includegraphics[scale=0.26]{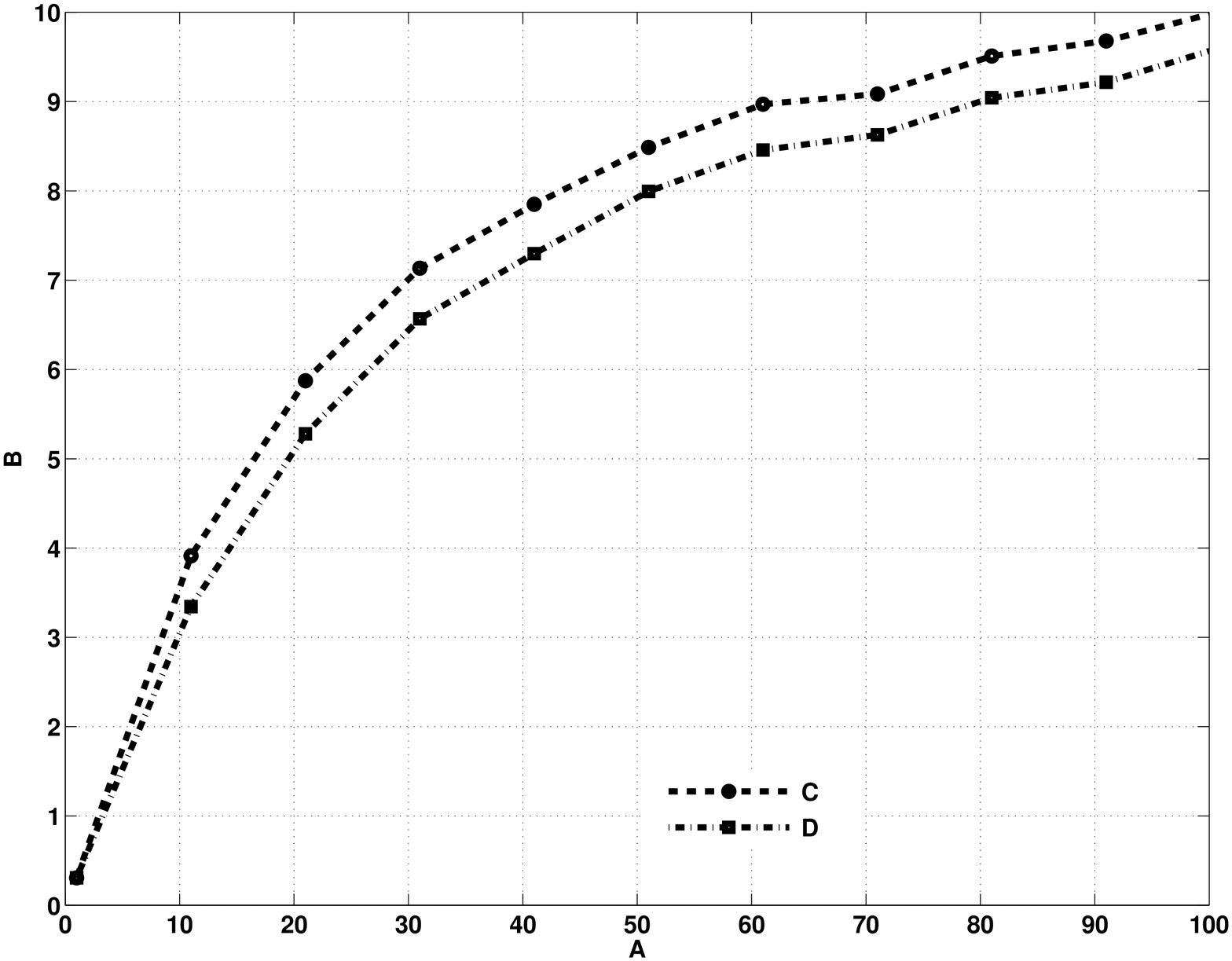}
	}
\subfloat[]{
  \includegraphics[scale=0.26]{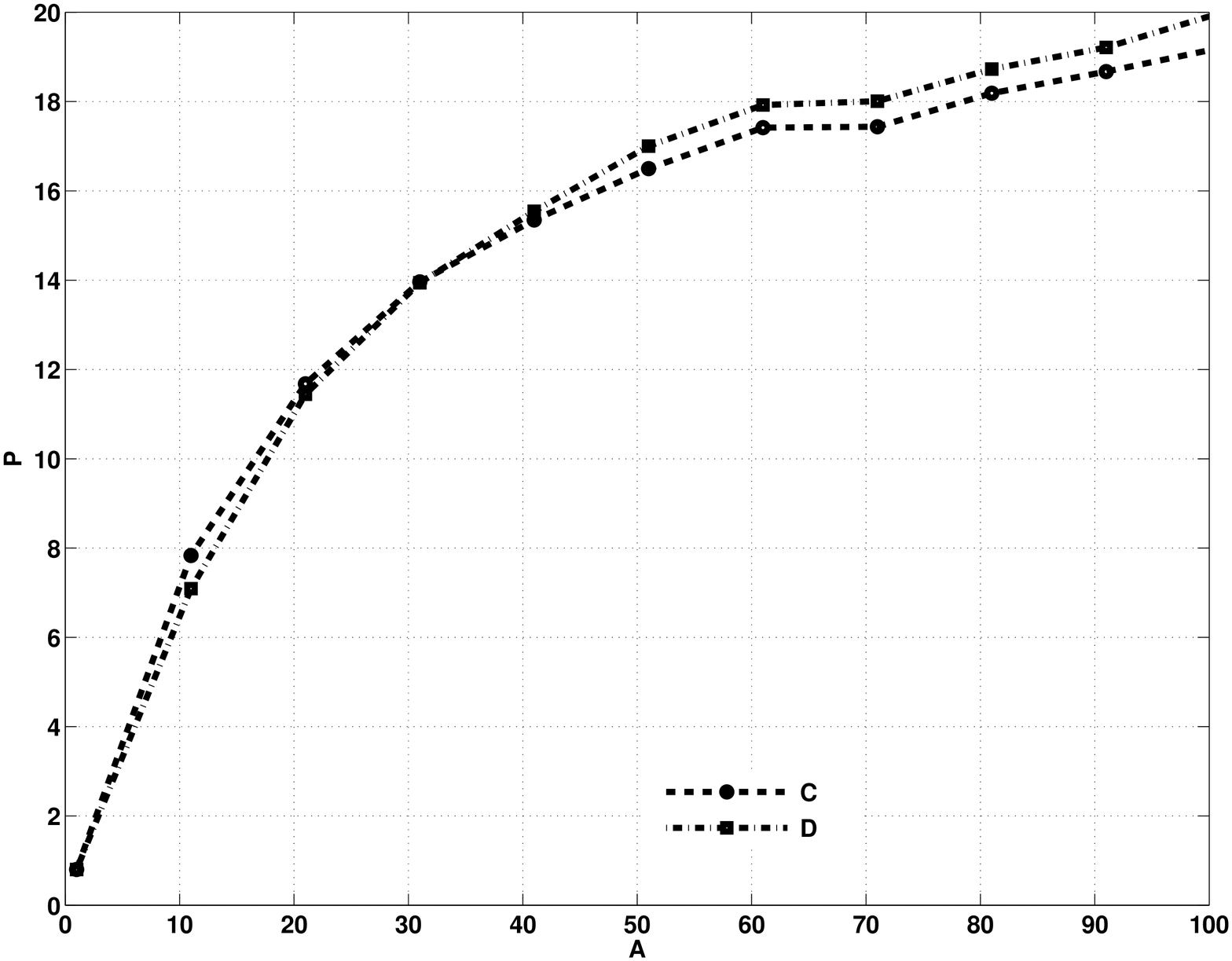}
	}
%
    \end{psfrags}  
\caption{Average throughput (a) and energy consumption per slot (b) comparison of throughput maximizing offline optimal policy  and online heuristic policy  against varying transmission window size for energy harvesting with memory.}
\label{avwthroughputmemory} 
\end{figure}
\begin{figure}[htpb]
\centering
  \begin{psfrags}
    \psfrag{A}[t]{Window Size (in number of slots)}
    \psfrag{B}[b]{Avg. Throughput (in Mbit/s)}
    \psfrag{C}[l]{\scriptsize{Offline Optimal}}
    \psfrag{D}[l]{\scriptsize{Online Heuristic}}
    \psfrag{E}[l]{\scriptsize{Power-Halving}}
    \includegraphics[scale=0.26]{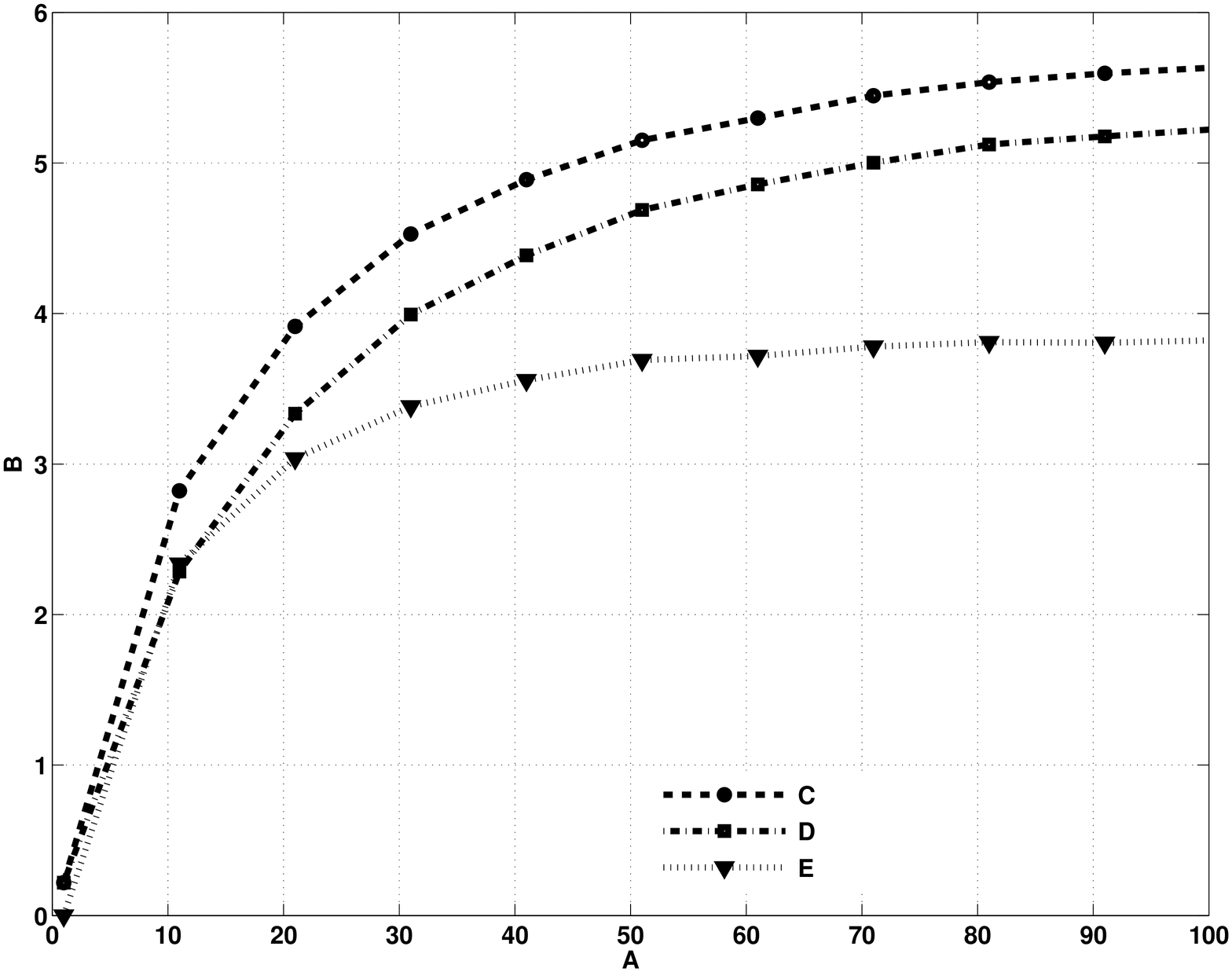}
    \end{psfrags}
\caption{Average throughput comparison of throughput maximizing offline optimal policy  and online heuristic policy  and Power-Halving policy  against varying transmission window size for stationary energy harvesting of $90$nJs occuring with $0.1$ probability in each time slot.}
\label{fig:maxthroughputcomppowerhalv}
\end{figure}
\section{Conclusion}
In this paper, we investigated finite horizon energy efficient transmission schemes in both offline and online problem settings. While the offline problem is a direct extention to existing offline problem formulations, our characterization of the offline optimal solution and the online approach that we introduce differ from previous studies as we intend to establish online optimality in relation with the statistical behavior of offline optimal transmission decisions. We believe these formulations and results could be useful also for similar problems where online performance over finite durations is crucial. In particular, the key contributions of this paper are the following:
\begin{itemize}
\item  In an offline setting, energy efficient transmission with a generic concave rate function is studied over a finite horizon considering energy and data arrivals as well as channel variations. The solution to the offline problem is formulated by offline optimal transmission decisions that depend only on future values  of energy harvests, data arrivals and channel variations. 
\item Based on the stochastic dynamic programming, the online optimal policy  is characterized as the policy that successively minimizes expected throughput losses with respect to the offline optimal transmission decisions.
\item To measure the efficiency of any online policy relative to the performance of the offline optimal transmission policy, the immediate fill metric is introduced. This metric can be also used to derive new online policies with performance guarantees as it can be lower bounded analytically.
\item Considering the simpler static channel case, the immediate fill of the policy that applies the expectation of offline optimal power level as the online power level is lower bounded and the distribution of offline optimal power level is characterized as the problem horizon approaches to infinity.
\end{itemize}

\section*{Acknowledgments}

This work was funded in part by TUBITAK and in part by the Science Academy of Turkey under a BAGEP award.

\section*{Appendix}
\subsection{The proof of Theorem \ref{optimalwater} }
\begin{proof}
We divide the proof of Theorem \ref{optimalwater} into two parts: 

(i) We show that if the water level of any slot $n$ is higher than the water level of the next slot $n+1$ ($w_{n} > \mathsf{w}_{n+1}$), then, there is an offline transmission schedule which achieves at least the same throughput or consumes at the most the same amount of energy with the initial schedule, i.e. the initial schedule with $w_{n} > \mathsf{w}_{n+1}$ for some slot $n$ is not an EE-OFF schedule. 

(ii) We show that in the offline optimal (EE-TM-OFF) policy, the water level $w_{n}$ is not lower than the maximum feasible level incurred by the inequalities resulting from the argument of part (i), i.e. $w_{n}=\min \lbrace w_{n}^{(e)}(w_{n}), w_{n}^{(b)}(w_{n}) \rbrace$ should be satisfied for any slot $n$ in an EE-TM-OFF policy.

Part (i):
Suppose that in a given transmission scheme $\pi$, $w_{n}>w_{n+1}$ for  some $n$. One can show that $\pi$ can be improved  by reducing $w_{n}$ and increasing $\mathsf{w}_{n+1}$ through one of the following: (Case a) move some data form slot $n$ to slot $n+1$ while keeping the total throughput achieved during $(n,n+1)$ fixed, (Case b) move some energy from slot $n$ to slot $n+1$ while keeping the total energy consumed during $(n,n+1)$ fixed.
Let $\rho_{n}^{\pi}$ and $\rho_{n+1}^{\pi}$ be the transmission power levels for slots $(n,n+1)$ belonging to the scheme $\pi$.

(Case a) Consider the following convex optimization problem for slots $(n,n+1)$:
\[
\min_{\rho_{n},\rho_{n+1}} \rho_{n}+\rho_{n+1}   
\]
\[
f(1+\rho_{n}\gamma_{n})+f(1+\rho_{n+1}\gamma_{n+1})= D_{n,n+1}
\]
\[
\rho_{n}\geq 0, \rho_{n+1}\geq 0
\]
where $D_{n,n+1}$ corresponds to the total throughput obtained  by the scheme $\pi$ during $(n,n+1)$, i.e. $f(1+\rho_{n}^{\pi}\gamma_{n})+f(1+\rho_{n+1}^{\pi}\gamma_{n+1})$.
The Lagrangian of the above problem can be written as follows:
\[
\mathcal{L} (\rho_{n},\rho_{n+1}, \lambda, \mu_{n}, \mu_{n+1} ) =
\]
\[
-(\rho_{n}+\rho_{n+1}) + \lambda (f(1+\rho_{n}\gamma_{n})+f(1+\rho_{n+1}\gamma_{n+1})-D_{n,n+1})
\]
\[
-\mu_{n}\rho_{n} - \mu_{n+1}\rho_{n+1}
\]
By setting $\dfrac{\partial\mathcal{L}}{\partial\rho_{n}}=0$, we get the following:
\[
\gamma_{n}f'(1+\rho_{n}\gamma_{n})=\frac{\mu_{n}+1}{\lambda} 
\]
Also, considering the complementary slackness for $\mu_{n}$, $\mu_{n}$ should be set to zero whenever  $\rho_{n}\geq 0$. Therefore, the optimal solution $\rho_{n}^{*}$ can be expressed as in the following:
\[
\rho_{n}^{*} = \frac{1}{\gamma_{n}}\left[ (f')^{-1}(\frac{1}{\lambda\gamma_{n}})-1\right]^{+} 
\]
Similarly, the optimal $\rho_{n+1}^{*}$ is as in below:
\[
\rho_{n+1}^{*} = \frac{1}{\gamma_{n+1}}\left[ (f')^{-1}(\frac{1}{\lambda\gamma_{n+1}})-1\right]^{+} 
\]
Accordingly, $(\rho_{n}+\rho_{n+1})$ is minimized when both water levels $w_{n}$ and $w_{n+1}$ are set to $\lambda$ that satisfies the total throughput constraint.

When $w_{n}>\mathsf{w}_{n+1}$, the optimal water level should be inside $(w_{n}, w_{n+1})$ as the total throughput  strictly decreasing with decreasing $w_{n}$ as long as $\rho_{n} > 0$ . Therefore, the water levels $w_{n}$ and $\mathsf{w}_{n+1}$ can always be equalized by transferring some data from slot $n$ to $n+1$ . This does not violate data causality  as the throughput at slot $n$ is reduced while the total throughput achieved during $(n,n+1)$ is preserved by increasing the throughput at slot $n+1$ to compensate. 

(Case b) Similarly, we consider the following optimization problem:
\[
\max_{\rho_{n},\rho_{n+1}}  f(1+\rho_{n}\gamma_{n})+f(1+\rho_{n+1}\gamma_{n+1})  
\]
\[
\rho_{n}+\rho_{n+1}  = E_{n,n+1}
\]
\[
\rho_{n}\geq 0, \rho_{n+1}\geq 0
\]
where $E_{n,n+1}$ corresponds to the total energy consumption  by the scheme $\pi$ during $(n,n+1)$, i.e. $E_{n,n+1}= \rho_{n}^{\pi}+\rho_{n+1}^{\pi}$.

The Lagrangian of the above problem can be written as follows:
\[
\mathcal{L} (\rho_{n},\rho_{n+1}, \lambda, \mu_{n}, \mu_{n+1} ) =
\]
\[
f(1+\rho_{n}\gamma_{n})+f(1+\rho_{n+1}\gamma_{n+1})  + \lambda ((\rho_{n}+\rho_{n+1}) -E_{n,n+1})
\]
\[
-\mu_{n}\rho_{n} - \mu_{n+1}\rho_{n+1}
\]
By setting $\dfrac{\partial\mathcal{L}}{\partial\rho_{n}}=0$, we get the following:
\[
\gamma_{n}f'(1+\rho_{n}\gamma_{n})=\mu_{n}-\lambda  
\]
After setting the KKT multiplier $\mu_{n}$ to zero where  $\rho_{n}\geq 0$, we get the following expression for the optimal  $\rho_{n}^{*}$:
\[
\rho_{n}^{*} = \frac{1}{\gamma_{n}}\left[ (f')^{-1}(\frac{\lambda}{\gamma_{n}})-1\right]^{+} 
\]
Similarly, for optimizing $\rho_{n+1}$ setting $\dfrac{\partial\mathcal{L}}{\partial\rho_{n+1}}=0$ gives the following expression for the optimal $\rho_{n+1}^{*}$:
\[
\rho_{n+1}^{*} = \frac{1}{\gamma_{n+1}}\left[ (f')^{-1}(\frac{\lambda}{\gamma_{n+1}})-1\right]^{+} 
\]
When both water levels $w_{n}$ and $\mathsf{w}_{n+1}$ are equalized to $\frac{1}{\lambda}$ that satisfies the total energy constraint, the total throughput achieved during the slots $(n,n+1)$ is maximized and this can be done whenever $w_{n}>w_{n+1}$ by transferring energy from $n$ and $n+1$ without violating energy causality or total energy constraints. Therefore in an EE-OFF schedule, water levels $w_n$s are non-decreasing with increasing $n$.

Part (ii):
By the energy causality, total energy consumption is bounded as follows:
\[
\displaystyle\sum_{l=n}^{n+u}\rho_{l} \leq e_{n}+ \displaystyle\sum_{l=n+1}^{n+u} H_{l}, u=1,2,....,(N-n),
\]
Expressing $\rho_{l}$ using  water levels:
\[
\displaystyle\sum_{l=n}^{n+u}\frac{1}{\gamma_{l}}\left[ (f')^{-1}(\frac{1}{w_{l}\gamma_{l}})-1\right]^{+} \leq e_{n}+ \displaystyle\sum_{l=n+1}^{n+u} H_{l}
\]
In an optimal scheme, $w_{n}\leq \mathsf{w}_{m}$ for any slot $m>n$  as it is proven in Part (i), thus:
\[
\displaystyle\sum_{l=n}^{n+u}\frac{1}{\gamma_{l}}\left[ (f')^{-1}(\frac{1}{w_{n}\gamma_{l}})-1\right]^{+} \leq \displaystyle\sum_{l=n}^{n+u}\frac{1}{\gamma_{l}}\left[ (f')^{-1}(\frac{1}{w_{l}\gamma_{l}})-1\right]^{+}
\]
And accordingly:
\[
\displaystyle\sum_{l=n}^{n+u}\frac{1}{\gamma_{l}}\left[ (f')^{-1}(\frac{1}{w_{n}\gamma_{l}})-1\right]^{+} \leq e_{n}+ \displaystyle\sum_{l=n+1}^{n+u} H_{l}
\]
The above inequality should be satisfied for any $u=1,2,....,(N-n)$ and it can be seen that $w_{n}$ is bounded by its lowest value for which the inequality holds with equality for some $u=1,2,....,(N-n)$. To find the energy bound value for $w_{n}$, the inequality can be transformed into the following form. 
\[
w_{n} \leq \frac{e_{n}+ \displaystyle\sum_{l=n+1}^{n+u} H_{l} + \displaystyle\sum_{l=n}^{n+u} K_{l}^{(e)}(w_{n}) }{u+1}
\]
The maximum value  of $w_{n}$ that satisfies the energy causality is given by the following: 
\[
w_{n}^{(e)}(w_{n})= \displaystyle\min_{u=0,...,(N-n)} \frac{e_{n}+ \displaystyle\sum_{l=n+1}^{n+u} H_{l} + \displaystyle\sum_{l=n}^{n+u} K_{l}^{(e)}(w_{n}) }{u+1} 
\]
Similarly, the data causality bounds the water level $w_{n}$ as follows:
\[
w_{n}^{(b)}(w_{n}) =\displaystyle\min_{v=0,...,(N-n)} \frac{b_{n}+ \displaystyle\sum_{l=n+1}^{n+v} B_{l} + \displaystyle\sum_{l=n}^{n+v} K_{l}^{(b)}(w_{n})  }{v+1}
\]
Any EE-TM-OFF schedule is EE-OFF by definition, hence $w_{n}\leq \min \lbrace w_{n}^{(e)}(w_{n}), w_{n}^{(b)} (w_{n})\rbrace$ for any EE-TM-OFF schedule. We will show that in EE-TM-OFF schedule water level $w_{n}$ also should not be smaller than $\min \lbrace w_{n}^{(e)}(w_{n}), w_{n}^{(b)} (w_{n})\rbrace$, i.e. $w_{n} \geq \min \lbrace w_{n}^{(e)}(w_{n}), w_{n}^{(b)}(w_{n}) \rbrace$.

Consider an EE-OFF schedule where $w_{n} = \min \lbrace w_{n}^{(e)}(w_{n}), w_{n}^{(b)} (w_{n})\rbrace$ for slot $n$ and $\mathsf{w}_{m}\leq \min \lbrace w_{m}^{(e)}(w_{m}), w_{m}^{(b)} (w_{m})\rbrace$ for slots $m>n$ since the schedule is EE-OFF. The selection of $w_{n}$ only affects the throughput achieved during the slots $n$ to $N$, hence if the reselection of $w_{n}$ as $w_{n} < \min \lbrace w_{n}^{(e)}(w_{n}), w_{n}^{(b)} (w_{n})\rbrace$ could improve the throughput achieved by the schedule within $[n,N]$ while keeping EE-OFF property, then the modified schedule could be EE-TM-OFF. This is not possible due to the observation in Remark \ref{derivativeremark}. When $w_{n} = \min \lbrace w_{n}^{(e)}(w_{n}), w_{n}^{(b)} (w_{n})\rbrace$, to improve the total throughput achieved in later slots $n+1, n+2, ....., N$, some energy/data can be moved from $n$ to later slots, however the throughput decrease in slot $n$ would be larger than the possible increase in some later slot $m>n$ as the derivative of the throughput with respect to power level (Remark \ref{derivativeremark}) decreases with increasing water level and $w_{m}\geq w_{n}$ in an EE-OFF policy. Hence, selecting the water level as $w_{n} = \min \lbrace w_{n}^{(e)}(w_{n}), w_{n}^{(b)} (w_{n})\rbrace$ always maximizes the total throughput as long as $w_{m}\geq w_{n}$ for $m>n$ which means all of the water levels $w_{m}$s after $n$ should be also selected as $w_{m} = \min \lbrace w_{m}^{(e)}(w_{m}), w_{m}^{(b)} (w_{m})\rbrace$.
\end{proof}
\subsection{The proof of Theorem \ref{thm:minmumfill} }
\begin{proof}
Consider the inequality for $n=N$:
\begin{equation}
 \label{lowerfillN}
 \eta^{w}(x_{N},\theta_{N}) \geq \displaystyle\min_{m \geq N} \displaystyle\min_{(x_{m},\theta_{m})} \mu_{m}^{w}(x_{m},\theta_{m})
\end{equation}
which means,
\[
 \eta^{w}(x_{N},\theta_{N}) \geq  \displaystyle\min_{(x_{N},\theta_{N})} \mu_{N}^{w}(x_{N},\theta_{N})
\]
The above inequality always holds as the offline optimal water level of the last slot $\tilde{w}_{N}^{*}$ is deterministic given $x_{N}$ implying  that $\eta^{w}(x_{N},\theta_{N})$ and $\mu_{N}^{w}(x_{N},\theta_{N})$ are both equal to $1$ if  $w_{N}=\tilde{w}_{N}^{*}$ for any $x_{N}$ and $\theta_{N}$.

Now, consider the following inequality:
\[
 \eta^{w}(x_{n+1},\theta_{n+1}) \geq \displaystyle\min_{m \geq n+1} \displaystyle\min_{(x_{m},\theta_{m})} \mu_{m}^{w}(x_{m},\theta_{m})
\]
We will show that the above inequality implies the inequality  (\ref{lowerfill}). The efficiency of the online policy $w$ can be expressed as follows:
\[
 \eta^{w}(x_{n},\theta_{n}) = 
\]
\[
\frac{(\log_{2}(w_{n}\gamma_{n}))^{+}+\displaystyle E_{X_{n}} [\hat{V}_{n+1 \mid \theta_{n+1}}^{*}(x_{n+1})\mid x_{n},\theta_{n}]}{(\log_{2}(w_{n}\gamma_{n}))^{+}+E_{\psi_{n}}[F_{n}(\tilde{w}_{n}^{*}, w_{n})+\tilde{V}_{n+1 \mid \theta_{n+1}}^{*}(x_{n+1})\mid x_{n},\theta_{n}]}
\]
\[
\geq \displaystyle\min \left\lbrace \mu_{n}^{w}(x_{n},\theta_{n}), \frac{\displaystyle E_{X_{n}} [\hat{V}_{n+1 \mid \theta_{n+1}}^{*}(x_{n+1})\mid x_{n},\theta_{n}]}{E_{\psi_{n}}[\tilde{V}_{n+1 \mid \theta_{n+1}}^{*}(x_{n+1})\mid x_{n},\theta_{n}]}\right\rbrace 
\]
\[
\geq \displaystyle\min \left\lbrace \mu_{n}^{w}(x_{n},\theta_{n}), \displaystyle\min_{(x_{n+1},\theta_{n+1})}  \eta^{w}(x_{n+1},\theta_{n+1})  \right\rbrace 
\]
\[
= \displaystyle\min_{m \geq n} \displaystyle\min_{(x_{m},\theta_{m})} \mu_{m}^{w}(x_{m},\theta_{m})
\]
Similarly, by the backward induction, the inequality (\ref{lowerfillN}) implies the inequality  (\ref{lowerfill}).
\end{proof}
\subsection{The proof of Proposition \ref{thm:immediateprop} }
\begin{proof}
To obtain the lower bound in Proposotion \ref{thm:immediateprop} for the immediate fill of the decision $\rho_{n}= w_{n}-1$, we first consider the immediate loss
term $E_{\psi_{m}}[F_{m}(\tilde{w}_{m}^{*}, w_{m})\mid x_{m},\theta_{m}]$ which is basically the expected throughput difference between the schedules$(\rho_{n},\tilde{\rho}_{n+1}^{\triangleright},....,\tilde{\rho}_{N}^{\triangleright})$ and $(\tilde{\rho}_{n}^{*},\tilde{\rho}_{n+1}^{*},...., \tilde{\rho}_{N}^{*})$ where $\tilde{\rho}_{n+1}^{\triangleright},....,\tilde{\rho}_{N}^{\triangleright}$ are offline optimal power levels following the decision $\rho_{n}$. The immediate loss is the expectation of the throughput difference in below:
\[
\log_2 (1+\tilde{\rho}_{n}^{*})-\log_2 (1+\rho_{n})+\xi(\rho_{n})
\]
where $\xi(\rho_{n})=\displaystyle\sum_{k=n+1}^{N}\log_2 (1+\tilde{\rho}_{k}^{*})-\displaystyle\sum_{k=n+1}^{N}\log_2 (1+\tilde{\rho}_{k}^{\triangleright})$

Then, we can bound the difference as follows:
\[\xi(\rho_{n}) = \displaystyle\sum_{k=n+1}^{N} \log_2\left( 1+ \frac{\tilde{\rho}_{k}^{*}-\tilde{\rho}_{k}^{\triangleright}}{1+\tilde{\rho}_{k}^{*}-(\tilde{\rho}_{k}^{*}-\tilde{\rho}_{k}^{\triangleright})}\right) 
\]
\[
\leq \displaystyle\max_{\Delta \in \mathbf{S}(\rho_{n}) } \displaystyle\sum_{k=n+1}^{N} \log_2\left( 1+ \frac{\Delta_k}{1+\tilde{\rho}_{k}^{*}-\Delta_k}\right)
\]
where $\Delta$ is the vector $[\Delta_{n+1},\Delta_{n+2},....,\Delta_{N}]$ and $\mathbf{S}(\rho_{n})$ is the set of all $\Delta$ vectors for which $\Delta$ is a possible instance of the vector $  [\tilde{\rho}_{n+1}^{*}-\tilde{\rho}_{n+1}^{\triangleright},\tilde{\rho}_{n+2}^{*}-\tilde{\rho}_{n+2}^{\triangleright},...., \tilde{\rho}_{N}^{*}-\tilde{\rho}_{N}^{\triangleright}]$. We know the following facts for any $\Delta$ vector in  the set $\mathbf{S}(\rho_{n})$: If $\Delta \in \mathbf{S}(\rho_{n})$,  $0 \leq \Delta \leq [\tilde{\rho}_{n+1}^{*},\tilde{\rho}_{n+1}^{*},....,\tilde{\rho}_{N}^{*}] $ and $\Vert \Delta \Vert_{1} = (\rho_{n}-\tilde{\rho}_{n}^{*})$ since the energy consumption of both $(\rho_{n},\tilde{\rho}_{n+2}^{\triangleright},....,\tilde{\rho}_{N}^{\triangleright})$ and $(\tilde{\rho}_{n}^{*},\tilde{\rho}_{n+1}^{*},...., \tilde{\rho}_{N}^{*})$ schedules should be equal. Now, consider the case $\rho_{n}< \tilde{\rho}_{n}^{*}$. Clearly, $\xi(\rho_{n})< 0$ for this case since the  offline optimal decisions $\tilde{\rho}_{k}^{\triangleright}$s have more energy to spend than the offline optimal decisions $\tilde{\rho}_{k}^{*}$s. Therefore, we can upper bound $\xi(\rho_{n})$ considering the instances of $\tilde{\rho}_{n}^{*}$ where $\rho_{n}\geq \tilde{\rho}_{n}^{*}$: 
\[
\xi(\rho_{n}) \leq \displaystyle\max_{\substack{\Delta \in \mathbf{S}(\rho_{n}) \\ \tilde{\rho}_{n}^{*} \leq \rho_{n} } }\displaystyle\sum_{k=n+1}^{N} \log_2\left( 1+ \frac{\Delta_k}{1+\tilde{\rho}_{k}^{\triangleright}}\right)\leq \displaystyle\max_{\substack{\Delta \in \mathbf{S}(\rho_{n}) \\ \tilde{\rho}_{n}^{*} \leq \rho_{n} } }\displaystyle\sum_{k=n+1}^{N} \log_2\left( 1+ \frac{\Delta_k}{1+\tilde{\rho}_{n+1}^{\triangleright}}\right)
\]
\[
\leq \displaystyle\max_{\substack{ \Vert \Delta \Vert_{1} = (\rho_{n}-\tilde{\rho}_{n}^{*}) \\ \tilde{\rho}_{n}^{*} \leq \rho_{n} } }\displaystyle\sum_{k=n+1}^{N} \log_2\left( 1+ \frac{\Delta_k}{1+\tilde{\rho}_{n+1}^{\triangleright}}\right)= (N-n)\log_2\left( 1+ \frac{\frac{(\rho_{n}-\tilde{\rho}_{n}^{*})}{N-n}}{1+\tilde{\rho}_{n+1}^{\triangleright}}\right), \tilde{\rho}_{n}^{*} \leq \rho_{n}
\]
\[
\leq \displaystyle\sup_{\substack{ N  \in \mathbb{N}^{+}\\ \tilde{\rho}_{n}^{*} \leq \rho_{n} } }(N-n)\log_2\left( 1+ \frac{\frac{(\rho_{n}-\tilde{\rho}_{n}^{*})}{N-n}}{1+\tilde{\rho}_{n+1}^{\triangleright}}\right)= \displaystyle\lim_{\substack{ N \rightarrow +\infty\\ \tilde{\rho}_{n}^{*} \leq \rho_{n} } }(N-n)\log_2\left( 1+ \frac{\frac{(\rho_{n}-\tilde{\rho}_{n}^{*})}{N-n}}{1+\tilde{\rho}_{n+1}^{\triangleright}}\right)
\]
\[
=\ln(2)\frac{\rho_{n}-\tilde{\rho}_{n}^{*}}{1+\tilde{\rho}_{n+1}^{\triangleright}}, \tilde{\rho}_{n}^{*} \leq \rho_{n}
\]

Therefore, $\xi(\rho_{n})$ can be upper bounded as:
\[
\xi(\rho_{n}) \leq \ln(2)\frac{(\rho_{n}-\tilde{\rho}_{n}^{*})^{+}}{1+\tilde{\rho}_{n+1}^{\triangleright}}
\]
since $\xi(\rho_{n})<0$ for $\rho_{n}< \tilde{\rho}_{n}^{*}$. Acccordingly,
\[
E_{\psi_{m}}[F_{m}(\tilde{w}_{m}^{*}, w_{m})\mid x_{m},\theta_{m}]\leq E[\log_2 (1+\tilde{\rho}_{n}^{*})]-\log_2 (1+\rho_{n})+ \ln(2)E\left[ \frac{(\rho_{n}-\tilde{\rho}_{n}^{*})^{+}}{1+\tilde{\rho}_{n+1}^{\triangleright}}\right] 
\]
Hence,
\[
\mu_{n}^{w}(x_{n},\theta_{n}) = \frac{\log_2(1+\rho_{n})}{\log_2(1+\rho_{n}) +E_{\psi_{m}}[F_{m}(\tilde{w}_{m}^{*}, w_{m})\mid x_{m},\theta_{m}] }
\]
\[
\geq \frac{\ln(1+\rho_{n})}{E\left[ \ln(1+\tilde{\rho}_{n}^{*})\right] +E\left[\frac{(\rho_{n}-\tilde{\rho}_{n}^{*})^{+}}{1+\tilde{\rho}_{n+1}^{\triangleright}}\right]  }
\]
\end{proof}
\subsection{The proof of Proposition \ref{immediatereprop} }
\begin{proof}
The bound in Proposition \ref{thm:immediateprop} can be simplified as follows,
\[
\mu_{n}^{w}(x_{n},\theta_{n}) \geq \frac{\ln(1+\rho_{n})}{E\left[ \ln(1+\tilde{\rho}_{n}^{*})\right] +E\left[ (\rho_{n}-\tilde{\rho}_{n}^{*})^{+}\right] }
\]

When $\rho_{n}=E[\tilde{\rho}_{n}^{*}]$, $E\left[ \ln(1+\tilde{\rho}_{n}^{*})\right] \leq \ln(1+E[\tilde{\rho}_{n}^{*}])$ due to Jensen's inequality. Therefore,
\[
\mu_{m}^{\check{w}}(x_{m},\theta_{m}) \geq \frac{1}{ 1 + \frac{E\left[(E[\tilde{\rho}_{n}^{*}]-\tilde{\rho}_{n}^{*})^{+}\right]}{\ln(1+E[\tilde{\rho}_{n}^{*}])}}
\]
\[
\geq \frac{1}{ 1 + \frac{E\left[\sqrt{(E[\tilde{\rho}_{n}^{*}]-\tilde{\rho}_{n}^{*})^{2}}\right]}{\ln(1+E[\tilde{\rho}_{n}^{*}])}}
\]
\[
\geq \frac{1}{ 1 + \frac{\sqrt{Var\left(\tilde{\rho}_{n}^{*}\right)}}{\ln(1+E[\tilde{\rho}_{n}^{*}])}}
\]
where the last step used Jensen's inequality on $\sqrt{.}$ function.
\end{proof}
\subsection{The proof of Theorem \ref{thm:ccdfrho} }
\begin{proof}

(i) Clearly, for any given $N$ and $e_n = x$, $\tilde{\rho}_{n}^{*}$ is lower bounded by $\frac{x}{N-n+1}$
which goes to infinity as $x$ goes to infinity hence the probability that $\tilde{\rho}_{n}^{*}$ is smaller than some $r$ should go to $0$. Similarly, $\tilde{\rho}_{n}^{*}$ is upper bounded by $x$ hence the probability that $\tilde{\rho}_{n}^{*}$ is smaller than some $r$ should go to $1$ as $x$ gets arbitrarily close to $r$.

(ii) The probability function $\Pr(\tilde{\rho}_{n}^{*}< \frac{h}{m} \mid e_n = x)$ can be interpreted as the probability that an energy outage occurs until the end of problem horizon when the power level $\frac{h}{m}$ energy/slot is continuously applied after the slot $n$ where the energy level is given as $e_n = x$.
 
By definition, for $x \leq \frac{h}{m}$, 
\[
\Pr(\tilde{\rho}_{n}^{*}< \frac{h}{m} \mid e_n = x) = 1
\]
For $x>\frac{h}{m}$, the energy outage does not occur at slot $n$. Therefore, if it occurs, the energy outage should occur after the slot $n$: 
\[
\Pr(\tilde{\rho}_{n}^{*}< \frac{h}{m} \mid e_n = x)=\Pr(\tilde{\rho}_{n+1}^{*}<\frac{h}{m} \mid e_{n+1} = x-\frac{h}{m}+H_{n})
\]
which means:
\begin{eqnarray}
&\Pr(\tilde{\rho}_{n}^{*}< \frac{h}{m} \mid \!\! e_n  \!\!  = x) =(1 \!\! -p)\Pr(\tilde{\rho}_{n+1}^{*}< \frac{h}{m} \mid \!\! e_{n+1}\!\! =x-\frac{h}{m})\nonumber \\ 
& +p\Pr(\tilde{\rho}_{n+1}^{*}< \frac{h}{m} \mid  \!\! e_{n+1} \!\! = x-\frac{h}{m}+h) &
\label{eq:cdfstretched}
\end{eqnarray}
For $N=n$,
\[
\Pr(\tilde{\rho}_{n}^{*}< \frac{h}{m} \mid e_n = x)= rect(\frac{mx}{h}-\frac{1}{2})
\]
Similarly, for $N=n+1$, Eq. (\ref{eq:cdfstretched}) gives the following as $\Pr(\tilde{\rho}_{n+1}^{*}<\frac{h}{m} \mid e_{n+1} = x)=rect(\frac{mx}{h}-\frac{1}{2})$:
\[
\Pr(\tilde{\rho}_{n}^{*}< \frac{h}{m} \mid e_n = x)=rect(\frac{mx}{h}-\frac{1}{2})+(1-p)rect(\frac{mx}{h}-\frac{3}{2})
\]
Now, suppose that:
\[
\Pr(\tilde{\rho}_{n+1}^{*}<\frac{h}{m} \mid e_{n+1} = x)=rect(\frac{mx}{h}-\frac{1}{2})+\displaystyle\sum_{j=1}^{K}a_{j,n+1}rect(\frac{mx}{h}-\frac{1}{2}-j)
\]
or assumming $a_{0,n+1}=1$,
\[
\Pr(\tilde{\rho}_{n+1}^{*}<\frac{h}{m} \mid e_{n+1} = x)=\displaystyle\sum_{j=0}^{K}a_{j,n+1}rect(\frac{mx}{h}-\frac{1}{2}-j)
\]

By Eq. (\ref{eq:cdfstretched}),
\[
\Pr(\tilde{\rho}_{n}^{*}<\frac{h}{m} \mid e_{n} = x)=rect(\frac{mx}{h}-\frac{1}{2})
\]
\[
+\displaystyle\sum_{j=1}^{K-m+1}((1-p)a_{j-1,n+1}+pa_{j+m-1,n+1})rect(\frac{mx}{h}-\frac{1}{2}-j)
\]
\[
+ \displaystyle\sum_{j=K-m+1}^{K}((1-p)a_{j-1,n+1}rect(\frac{mx}{h}-\frac{1}{2}-j)
\]
or
\[
\Pr(\tilde{\rho}_{n}^{*}<\frac{h}{m} \mid e_{n} = x)= \displaystyle\sum_{j=0}^{K+1}a_{j,n}rect(\frac{mx}{h}-\frac{1}{2}-j)
\]
where $a_{0,n}=1$, $a_{j,n}=((1-p)a_{j-1,n+1}+pa_{j+m-1,n+1})$ for $j=1,2,......,K-m+1$ and $a_{j,n}=(1-p)a_{j,n+1}$ for $j=K-m+2,......,K$.

By induction, it can be seen that $K=N-n$:
\[
\Pr(\tilde{\rho}_{n}^{*}<\frac{h}{m} \mid e_{n} = x)= \displaystyle\sum_{j=0}^{N-n+1}a_{j,n}rect(\frac{mx}{h}-\frac{1}{2}-j)
\]
As $N \rightarrow +\infty$, $a_{j,n+1} \rightarrow a_{j,n}$, the reccurence relation $a_{j,n} = (1-p)a_{j-1,n} + p a_{j+m-1,n}$ should hold and it can be satisfied  when $a_{j,n}=\Phi(m)^{j}$ where:
\[
p\Phi(m)^{m}-\Phi(m)+1-p=0
\]
Accordingly,
\[
\displaystyle\lim_{ N \rightarrow +\infty }\Pr(\tilde{\rho}_{n}^{*}<\frac{h}{m} \mid e_{n} = x)= \displaystyle\sum_{j=0}^{\infty}\Phi(m)^{j}rect(\frac{mx}{h}-\frac{1}{2}-j)
\]
or
\[
\displaystyle\lim_{ N \rightarrow +\infty } \Pr(\tilde{\rho}_{n}^{*}< \frac{h}{m} \mid e_n = x) = \Phi(m)^{\lfloor\frac{mx}{h}\rfloor} 
\]
Note that the equation $p\Phi(m)^{m}-\Phi(m)+1-p=0$ always has a root at $\Phi(m)=1$. However, if the equation has a root in $(0,1)$, $\Phi(m)$ should be equal that value since $0<a_{j,n}<1$ unless $a_{j-1,n}=1$ and $a_{j+m-1,n}=1$. Hence, $\Phi(m)$ is the minimum root of $p\Phi(m)^{m}-\Phi(m)+1-p=0$ in $(0,1]$.
\end{proof} 



\bibliographystyle{IEEEtran}
%

\bibliography{LazySched_BlackSeaComm} 

\end{document}